\documentclass{IEEEoj}
\usepackage{fix-cm}
\usepackage[compatibility=false]{caption}
\usepackage{cite}
\usepackage{amsmath,amssymb,amsfonts}
\usepackage{amsthm}
\usepackage{algorithmic}
\usepackage{graphicx,color}
\usepackage{subfig}
\usepackage{xcolor}%
\usepackage{hyperref}
\usepackage{textcomp}
\usepackage{multirow}
\usepackage{booktabs}
\def\BibTeX{{\rm B\kern-.05em{\sc i\kern-.025em b}\kern-.08em
    T\kern-.1667em\lower.7ex\hbox{E}\kern-.125emX}}
\AtBeginDocument{\definecolor{ojcolor}{cmyk}{0.93,0.59,0.15,0.02}}
\def\OJlogo{\vspace{-4pt}\hskip-4pt\includegraphics[height=18pt]{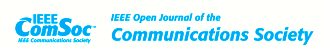}}

\newtheorem{proposition}{Proposition}
\newtheorem{corollary}{Corollary}

\begin{document}

\receiveddate{17 September, 2026}
\reviseddate{15 October, 2026}
\accepteddate{9 November, 2026}
\currentdate{9 November, 2026}

\title{A Knowledge-Centric Communication For Autonomous Cislunar Networks}

\author{AFAN ALI\IEEEauthorrefmark{1} \IEEEmembership{(Senior Member, IEEE)}, DANIEL BENEVIDES DA COSTA\IEEEauthorrefmark{1} \IEEEmembership{(Senior Member, IEEE)}, AND ALI ARSHAD NASIR\IEEEauthorrefmark{1} \IEEEmembership{(Senior Member, IEEE)}}
\affil{Interdisciplinary Research Center for Communication Systems and Sensing (IRC-CSS), Department of Electrical Engineering, King Fahd University of Petroleum and Minerals (KFUPM), Dhahran 31261, Saudi Arabia}
\markboth{IEEE Open Journal of the Communications Society}{Ali \textit{et al.}: A Knowledge-Centric Communication For Autonomous Cislunar Networks}

\begin{abstract}
Future lunar infrastructure requires communication as a persistent service rather than a mission-specific capability. Communication, digital twin, and artificial intelligence (AI) research have advanced autonomy independently while largely assuming that system state can eventually be reconciled with ground truth. Cislunar communication violates this assumption because propagation delay is fundamentally limited by the speed of light, visibility is governed by orbital geometry, and autonomous decisions often precede confirmatory observations. Here we introduce a knowledge-centric communication framework in which a digital twin continuously integrates delayed observations, communication physics, learned models, uncertainty quantification, and mission objectives into an evolving Operational Knowledge State. To quantify what is known, how uncertain it remains, and how current it is, we define two quantities over the fused Bayesian belief maintained for each communication link: \emph{Knowledge Entropy}, the differential entropy of the posterior belief, and \emph{Knowledge Freshness}, an Age of Information (AoI)-derived clock associated with the most recently assimilated evidence. Using standard Gaussian information-form fusion and covariance prediction as the inference layer, we establish the consistency conditions required by the proposed framework and demonstrate the operational significance of the resulting belief, entropy, and freshness through simulations and comparisons with flight data from the Longjiang-2 lunar microsatellite.
\end{abstract}


\maketitle

\section{INTRODUCTION}
\IEEEPARstart{H}{uman} exploration of the Moon is entering a different operational era. The Apollo missions relied on short-duration expeditions supported by mission-specific communication systems. In contrast, current international programs envision a persistent cislunar ecosystem comprising crewed habitats, robotic explorers, scientific observatories, commercial services, and in-situ resource utilization. The National Aeronautics and Space Administration (NASA) Artemis program, the Lunar Gateway, the European Space Agency (ESA) Moonlight initiative, LunaNet, and the International Lunar Research Station together point toward continuously operating space infrastructure in which communication is a persistent operational service rather than a mission-specific capability \cite{NASA_CPNT,LunaArchitecture2026}.

Delivering communication as a persistent service, rather than a mission-specific one, has required its own technical advances. Early work on Delay/Disruption Tolerant Networking (DTN) established the networking principles needed to operate under long propagation delays and intermittent connectivity \cite{schlesinger_delaydisruption_2017}. Optical communication later showed that laser links can substantially increase deep-space data rates over conventional radio-frequency (RF) systems, an approach validated in operation by NASA's Deep Space Optical Communications (DSOC) project \cite{velasco_operational_2026,DSOC2026}. NASA's Communication, Positioning, Navigation and Timing (CPNT) architecture, LunaNet, and emerging interdomain lunar communication frameworks have since defined interoperable communication services for cislunar operations \cite{NASA_CPNT,cetin_advancing_2026}. Communication research has also expanded from propagation modeling and link-budget analysis toward semantic communications, intelligent routing, and AI-assisted systems that adapt to changing communication environments \cite{fernandez_assessing_2020,wei_semantic-empowered_2026,LunaArchitecture2026}.

\begin{figure*}[t]
  \centering
    \includegraphics[width=1.3\columnwidth]{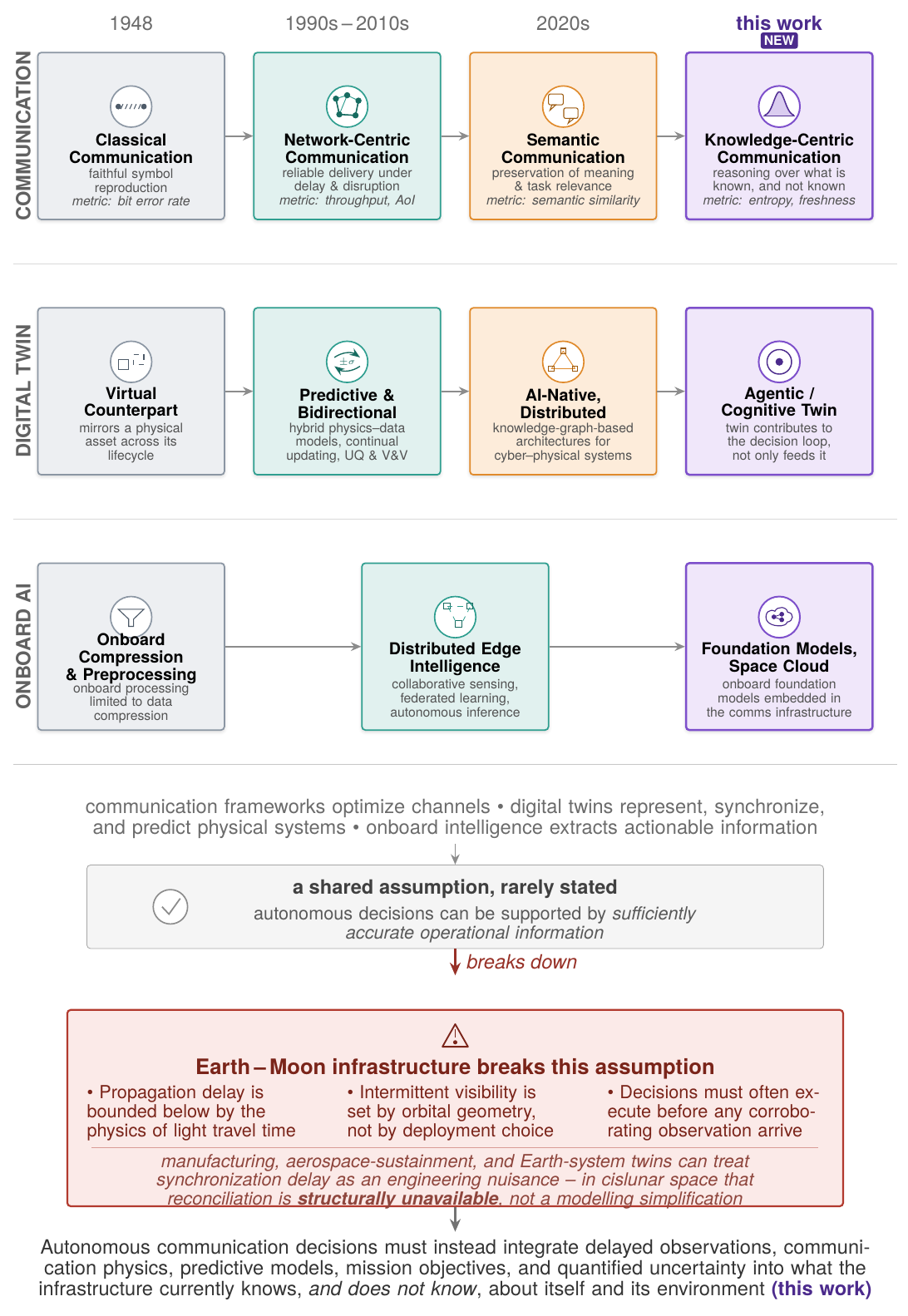}
  \caption{{Three research threads converging on autonomy.} Communication theory, digital twins, and onboard AI independently converge toward autonomy while assuming eventual ground-truth reconciliation.}
  \label{fig1}
\end{figure*}

Adapting communication systems to changing environments is, in a longer view, only the latest step in a series of shifts in what communication theory takes as its object. The classical paradigm established by Shannon's mathematical theory of communication treated communication as the faithful reproduction of transmitted symbols independent of their meaning \cite{shannon1948}. A network-centric paradigm later shifted attention to routing, topology, and delay-tolerant delivery and introduced timeliness metrics such as AoI \cite{kaul2012realtime}. More recently, a context-aware or semantic paradigm has argued that communication should preserve a message's meaning or task relevance rather than its exact symbol sequence \cite{gunduz2023beyond,qin2021semantic}. Each paradigm has produced its own measurable quantities, including bit error rate (BER), throughput, AoI, semantic similarity, and task success, as summarized in Fig.~\ref{fig1}(top row). Digital twin research has undergone a parallel series of shifts from virtual counterparts that represent a physical asset across its lifecycle \cite{Tao2019} to predictive, bidirectional decision-support systems that combine hybrid physics and data models with continual updating, uncertainty quantification, and verification and validation \cite{Willcox2024Comment,Ferrari2024}. More recent developments include knowledge-graph-based, AI-native, and distributed digital twin architectures for complex cyberphysical systems (Fig.~\ref{fig1}, middle row) \cite{san_evolution_2026,zhou_digital_2026,hamwi_fl-twin_2026,li_exploring_2026,parwez_digital_2026,wang_digital_2026,duran_toward_2026}. This trajectory from passive representation toward active reasoning is now described as a shift toward agentic or cognitive digital twins, in which the twin contributes directly to the decision loop \cite{san_evolution_2026}. The autonomous operation supported by these platforms is enabled by a parallel shift in AI from onboard processing toward distributed edge intelligence, collaborative sensing, federated learning, semantic reasoning, and autonomous onboard inference. This shift is often described as a move from Earth observation to Earth action, in which spacecraft interpret observations onboard rather than transmitting raw measurements to ground stations \cite{barretta_toward_2026}. Space cloud architectures, onboard foundation models, distributed AI, and semantic communication point in the same direction. Future space systems will rely on onboard intelligence embedded within the communication infrastructure rather than only on ground-based processing (Fig.~\ref{fig1}, bottom row) \cite{barretta_toward_2026,wei_semantic-empowered_2026,lin_embracing_nodate}.
Communication, digital twin, and AI research have each, then, been converging on autonomy from a different starting point, but largely without reference to one another, and so they answer different questions. Communication frameworks optimize channels, routing, coding, and resource allocation. Digital twins focus on representing, synchronizing, and predicting physical systems. Onboard intelligence concentrates on extracting actionable information from sensed data. Despite these differences, all three share a common assumption that autonomous decisions can be supported by sufficiently accurate operational information (Fig.~\ref{fig1}).

\section{MOTIVATION: THE GROUND-TRUTH ASSUMPTION BREAKS DOWN IN CISLUNAR COMMUNICATION}
Future Earth to Moon infrastructure breaks this assumption in a way that is qualitatively different from the domains in which digital twins have matured. Digital twins in manufacturing, aerospace sustainment, and Earth systems have generally treated synchronization delay as an engineering nuisance to be minimized because ground-based sensing, near-real-time telemetry, or offline reprocessing eventually supplies a trustworthy reference state against which the twin can be validated \cite{Willcox2024Comment,Ferrari2024}. However, Earth to Moon communication removes this option because propagation delay is bounded below by the physics of light travel time. Additionally, intermittent visibility is dictated by orbital geometry rather than deployment choices, and autonomous decisions must often be executed before any corroborating observation can arrive. The assumption shared by these prior domains, that a subsystem can eventually reconcile its internal state against ground truth, is therefore not a modeling simplification in cislunar communication but is structurally unavailable. Long propagation delays, intermittent line-of-sight (LoS) visibility, dynamic orbital geometries, heterogeneous communication technologies, sparse sensing, changing resource availability, and shifting mission priorities mean that no subsystem can maintain a complete and continuously synchronized picture of the operational environment. Therefore, autonomous communication decisions cannot rely on estimates of instantaneous physical state alone. They must instead integrate delayed observations, communication physics, predictive models, historical information, mission objectives, and quantified uncertainty into a representation of what the infrastructure currently knows and does not know about itself and its environment. Reasoning over such a representation is useful only if what is known and how current that knowledge remains can be measured. This motivates two metrics, which we call \textbf{Knowledge Entropy} and \textbf{Knowledge Freshness}. We define Knowledge Entropy as the entropy of the approximate posterior distribution that the twin maintains over the unobserved components of infrastructure state. It quantifies how much remains unknown. Similarly, Knowledge Freshness is an AoI-derived quantity \cite{kaul2012realtime} that measures how long ago the evidence underlying a belief was collected and, therefore, how far it may have decayed toward staleness.

AoI, age of incorrect information (AoII), value of information (VoI), and filtering methods can all be defined for multisource settings. The distinction claimed here is therefore not the number of sources. Rather, the Knowledge-Centric Digital Twin (KCDT) attaches evidence age and posterior entropy to every belief in a time-varying graph and exposes that persistent fused state to mission-specific decision rules. Consequently, a state-centric belief that freezes its uncertainty between synchronization reports can remain misleadingly confident, whereas the KCDT belief grows appropriately more uncertain between observations and contracts, though rarely to zero, when a delayed observation is fused. To ground this behavior in simulation rather than illustration alone, we model a representative cislunar link comprising the relay to Gateway hop of a three-segment chain from Earth through a relay and the Gateway to the lunar surface. It combines two-body orbital propagation, a closed-form path-loss and Doppler model, and a stochastic residual channel that represents the multipath and scintillation effects that the physics model cannot capture. Intermittent visibility and delayed, incomplete observations follow directly from this geometry and provide a concrete testbed for the knowledge-state behavior introduced above. Figure~\ref{fig2} traces Knowledge Entropy in natural units of information, or nats, over a representative 72-hour relay to Gateway communication link. The state-centric curve remains nearly constant, whereas the knowledge-centric curve increases progressively during communication outages as uncertainty accumulates. Whenever a delayed observation becomes available, it is incorporated into the Operational Knowledge State through Bayesian fusion, which causes entropy to decrease before growing again if no further observations arrive. Thus, the reductions are determined by delayed observation arrivals arising from the communication geometry rather than by periodic updates. Operational Knowledge continues to evolve even while the link remains silent.

\begin{figure}[t]
  \centering
    \includegraphics[width=1\columnwidth]{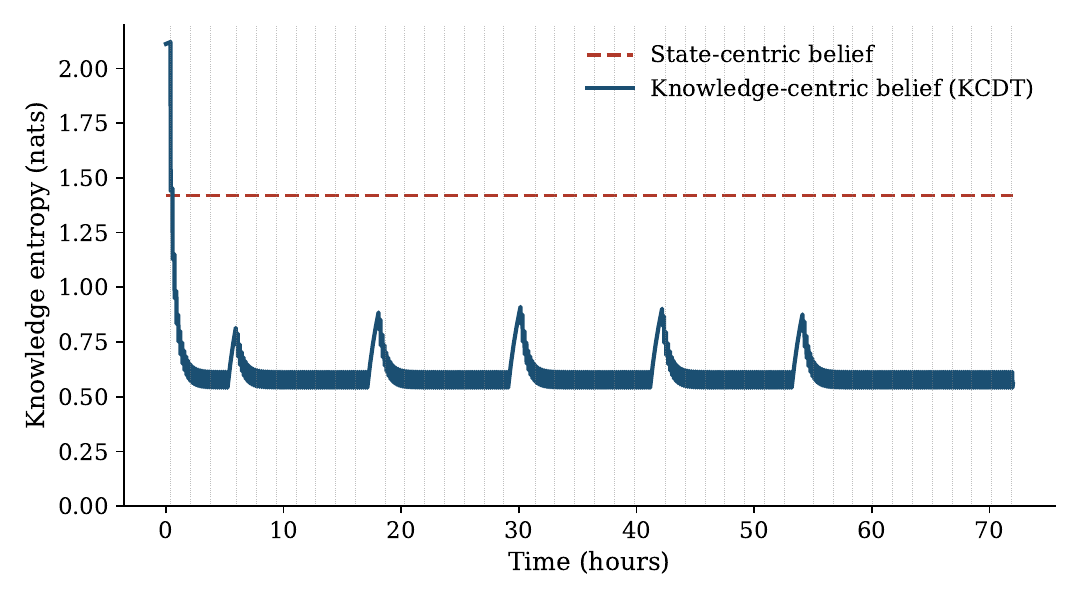}
  \caption{Knowledge Entropy under realistic delayed and intermittent observation for the Regime~B relay to Gateway link over a 72-hour window.}
  \label{fig2}
\end{figure}

The reason is that the two representations are answering different questions. State-centric tracking assumes the last observation remains a valid description of the communication channel until a new observation becomes available, so its reported uncertainty remains nearly constant at approximately 1.4~nats regardless of how long the communication outage or observation gap persists. The knowledge-centric twin does not make such an assumption. Instead, the entropy increases steadily from its post-update value to roughly 2.1~nats during the longest observation gap in the trace before decreasing again when a delayed observation is incorporated through Bayesian fusion. One detail worth emphasizing is that the entropy never returns to zero after an update. Each newly received observation resolves only the uncertainty supported by the available evidence, while the residual uncertainty arising from incomplete coverage, propagation delay, and channel randomness remains represented within the Operational Knowledge State.

However, these two metrics vary in their significance across the infrastructure in different cases, since what must be known precisely, and how fresh that knowledge must remain, is set by the mission itself. For example, a safety-critical handover during a robotic rover's lunar excursion may tolerate almost no entropy in the relevant link's state, while a background telemetry channel for a dormant instrument may tolerate substantial staleness, so that mission objectives act on the knowledge state as a requirement-setting function, specifying, for each element of the infrastructure, the entropy and freshness that autonomous decisions require, rather than demanding uniform precision throughout.

Existing communication, digital twin, and AI research typically addresses different parts of this problem separately, leaving open how autonomous space infrastructure should represent, update, and reason over operational knowledge when communication, navigation, sensing, and mission objectives are simultaneously delayed, incomplete, heterogeneous, and uncertain. We answer this question by moving beyond communication optimization or digital twin synchronization toward a framework in which operational knowledge itself, rather than symbols, channels, or meaning alone, becomes the object of autonomous reasoning. We refer to this shift as the \emph{Knowledge Age of Communications}. Concretely, we argue that future Earth to Moon communication systems should move from state-centric communication to knowledge-centric reasoning realized through the KCDT. The KCDT continuously fuses communication physics, delayed observations, learned models, mission objectives, and quantified uncertainty into a single evolving knowledge state. Earth to Moon communication motivates the framework, but it applies to any autonomous cyberphysical infrastructure whose reliable decisions depend on evolving knowledge rather than instantaneous observation.

A framework for Earth to Moon communication is depicted in Fig.~\ref{fig3a} and Fig.~\ref{fig3b}. Fig.~\ref{fig3a} shows the modern cislunar ecosystem that motivates it. This heterogeneous and time-varying infrastructure spans terrestrial ground stations, Earth-orbit satellites labeled SAT1 and SAT2, lunar-orbit satellites, robotic rovers on the Moon, and low Earth orbit (LEO) satellites. Long propagation delays, intermittent or obstructed visibility, and resource-constrained links mean that no single link is permanently available. The Earth-side portion of such a chain can be embedded within a broader architecture that integrates satellite, aerial, and terrestrial networks. Recent studies have addressed specific objectives in these integrated networks, including reliable covert communication under nonorthogonal multiple access (NOMA) and spatially distributed jamming, privacy-preserving federated learning, and hierarchical resource optimization~\cite{guo_reliable_2026,JSAC2025correlated,TWC2026correlated}. These mechanisms optimize a specified communication, security, learning, or allocation objective using the state information available to the corresponding controller. The KCDT addresses the complementary upstream question of what operational state can be trusted when observations are delayed, incomplete, heterogeneous, and uncertain. The resulting link beliefs can subsequently be supplied to such task-specific mechanisms. Fig.~\ref{fig3b} shows how the KCDT turns this environment into an Operational Knowledge State by fusing four heterogeneous streams. Delayed observations describe a past rather than a present state. Communication physics constrains what that state can plausibly be independent of any single observation. Learned models capture the residual behavior that physics cannot explain together with their own predictive uncertainty. Mission objectives specify how precisely each part of the environment must be known. This representation keeps the knowledge state suitable for the decisions it supports rather than exhaustively detailed. It separates confident knowledge from uncertain knowledge and supports decisions calibrated to that confidence, including adaptive routing, proactive link switching, uncertainty-hedged resource allocation, and operator-facing analytics. Section~\ref{sec:decisioninterface} formalizes representative routing, link-switching, and resource-allocation rules that consume the maintained knowledge state. These rules provide a concrete interface between estimation and control. They are reference policies rather than claims of optimality for every cislunar mission. The resulting actions ultimately close the loop with the physical infrastructure. Fig.~\ref{fig3b} therefore represents a continuous cycle in which knowledge, rather than instantaneous state, is what the system produces and acts on.

\begin{figure*}[t]
  \centering
  \includegraphics[width=1.5\columnwidth]{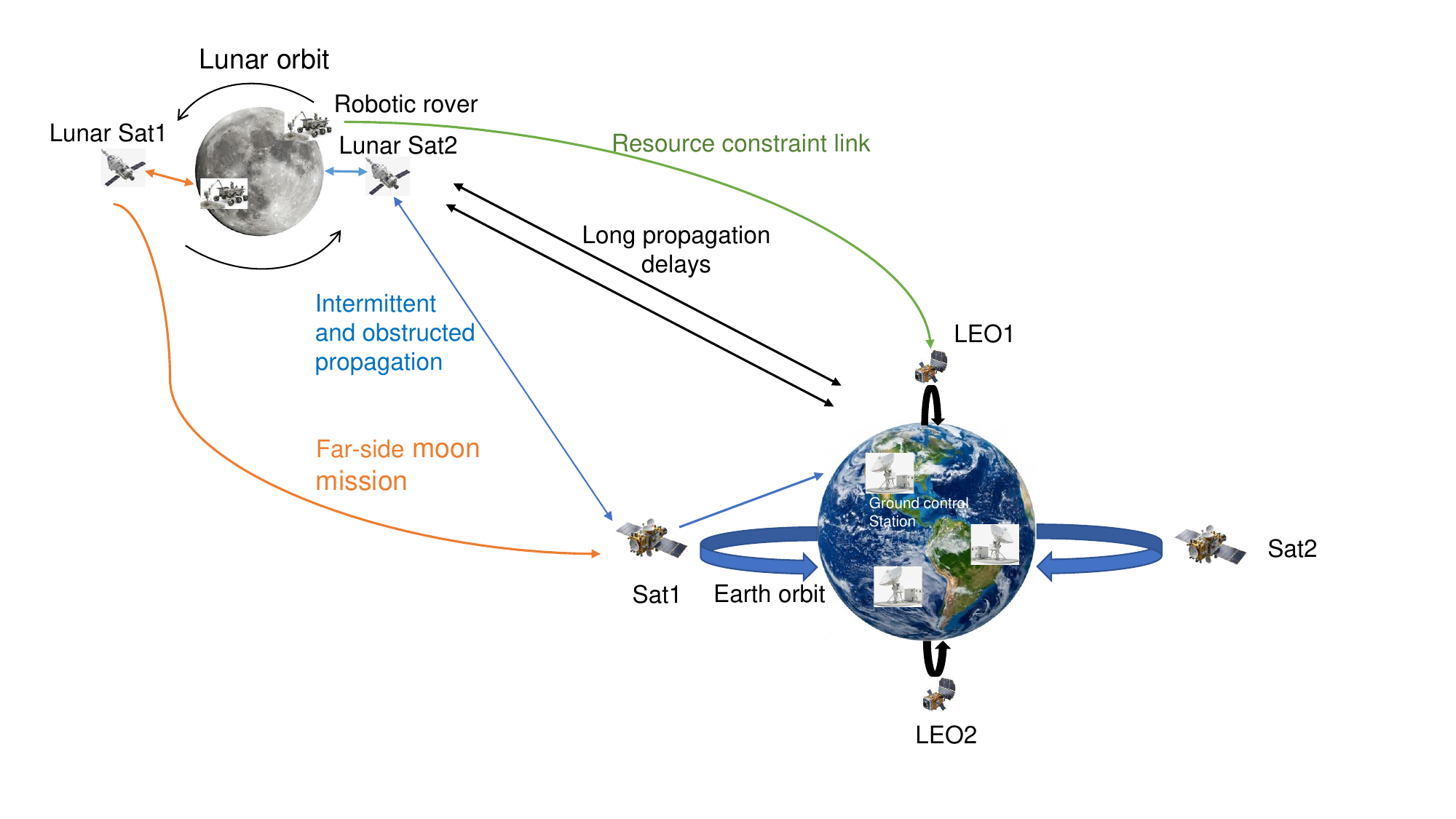}
    \caption{The cislunar ecosystem, comprising terrestrial ground stations, relay satellites, the lunar gateway, and robotic explorers, represented as a dynamic graph subject to intermittent visibility, resource constraints, and heterogeneous communication technologies.}
    \label{fig3a}
\end{figure*}

\begin{figure*}[t]
\centering
    \includegraphics[width=1.3\columnwidth]{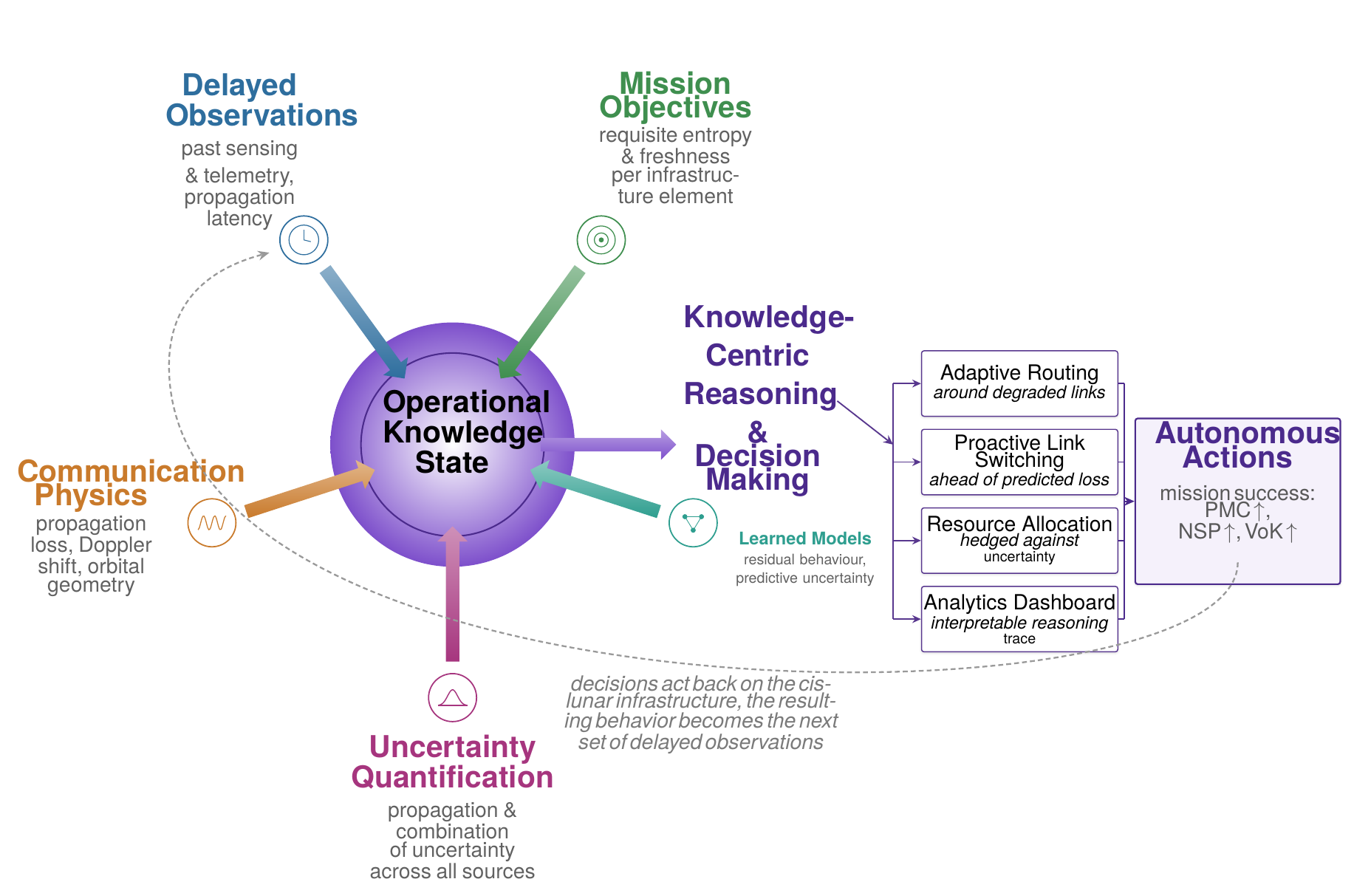}
    \caption{The knowledge-centric digital twin continuously fuses delayed observations, communication physics, learned models, mission objectives, and uncertainty quantification into a single evolving operational knowledge state.}
     \label{fig3b}
\end{figure*}

This distinction is not merely representational, it changes operation of the system as soon as a link falls silent. Figure~\ref{fig4} makes this concrete by following the same observation gap under each paradigm, side by side. A state-centric system pauses, i.e., its belief is frozen at the last observation, it reports unchanged confidence throughout the silence, and it is caught unprepared when a new observation finally arrives and reveals how much may have drifted. A knowledge-centric system instead keeps reasoning through the gap, explicitly tracking how its uncertainty grows and continuing to produce calibrated decisions, such as switching links proactively or hedging resource allocation, so that the eventual new observation refines an already-appropriate belief rather than forcing an abrupt correction. Consequently, when observations stop, communication should not stop reasoning, and Fig.~\ref{fig4} depicts this work's argument for why that difference is fundamental rather than cosmetic.

\begin{figure*}[t]
\centering
\includegraphics[width=1\textwidth]{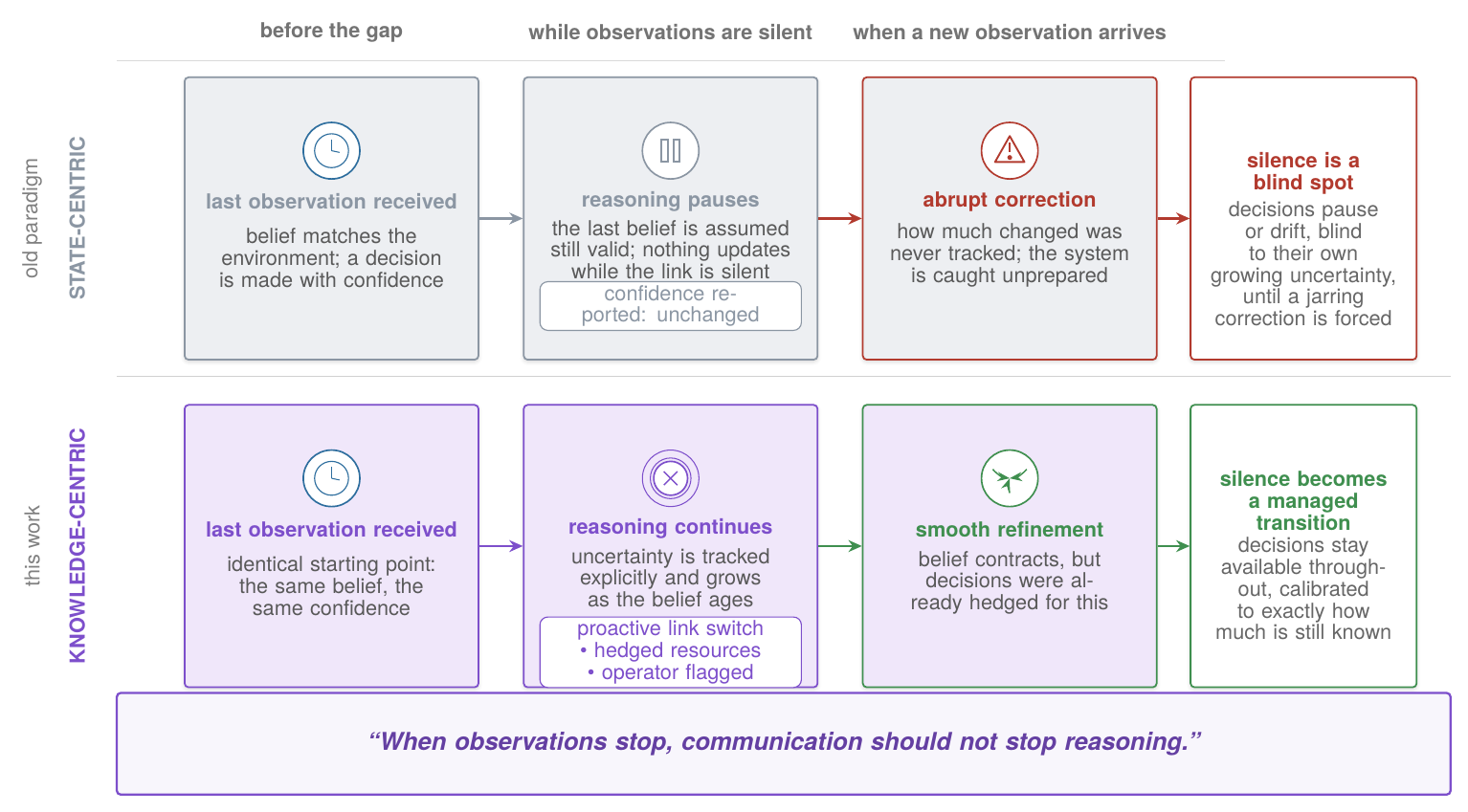}
\caption{\textbf{When observations stop, communication should not stop reasoning.} The same silent link, reasoned about two different ways, both starting from an identical belief. Under state-centric reasoning (top), belief freezes for the duration of the gap and a new observation forces an abrupt correction. Under knowledge-centric reasoning (bottom, this work), uncertainty is tracked explicitly as it grows, decisions stay available and hedged throughout, and the new observation smoothly refines an already-appropriate belief. Silence is a blind spot in the first case and a managed transition in the second.}
\label{fig4}
\end{figure*}

\section{SYSTEM MODEL}

\subsection{CISLUNAR COMMUNICATION ENVIRONMENT AND DYNAMIC GRAPH REPRESENTATION}

We represent the cislunar communication infrastructure as a time-varying graph, which can be written as

\begin{equation}
G(t)=\left(V(t),E(t)\right),
\end{equation}
where $V(t)$ is the set of communication nodes at time $t$, comprising of terrestrial ground stations, relay satellites, the Lunar Gateway, robotic explorers, crewed lunar habitats, and $E(t)\subseteq V(t)\times V(t)$ is the set of active communication links. Node positions change with orbital motion, so both the topology of $G(t)$ and the physical characteristics of each link change with it.

Earth-orbiting relay satellites in Regime~A are propagated with the classical two-body Keplerian model~\cite{vallado2013}. The Lunar Gateway in Regime~B is propagated with a reduced-order approximation of its Near Rectilinear Halo Orbit (NRHO) because two-body dynamics do not reproduce NRHO motion~\cite{Ferrari2024}. A full Circular Restricted Three-Body Problem (CR3BP) integration was beyond the scope of this validation. We return to this simplification in Section~\ref{sec:discussion}.

For every ordered node pair $(i,j)$, LoS visibility follows from a geometric sphere-occlusion test against the Earth or Moon, giving a binary indicator

\begin{equation}
v_{ij}(t)\in\{0,1\},
\end{equation}
where $v_{ij}(t)=1$ means an unobstructed LoS path exists between nodes $i$ and $j$, and $v_{ij}(t)=0$ means the link is blocked by planetary occlusion. A communication edge exists only when

\begin{equation}
v_{ij}(t)=1,
\end{equation}
so communication opportunities follow directly from orbital geometry, not from an assumed schedule.

Every active link $(i,j)$ carries a \emph{true communication state}, which represents the physical condition of the channel at time $t$. Before introducing its mathematical representation, we define the signal-to-noise ratio (SNR). The state can then be written as

\begin{equation}
\mathbf{s}_{ij}(t)
=
\left[
\mathrm{SNR}_{ij}(t),
f^{D}_{ij}(t),
\tau^{p}_{ij}(t)
\right]^{\mathrm T},
\label{eq:state}
\end{equation}
where $\mathrm{SNR}_{ij}(t)$ is the received SNR, $f^{D}_{ij}(t)$ is the Doppler shift from relative radial motion, and $\tau^{p}_{ij}(t)$ is the one-way propagation delay.

The one-way delay follows directly from separation, which can be defined as

\begin{equation}
\tau^{p}_{ij}(t)
=
\frac{R_{ij}(t)}{c},
\end{equation}
where

\begin{equation}
R_{ij}(t)
=
\left\|
\mathbf{r}_{i}(t)
-
\mathbf{r}_{j}(t)
\right\|
\end{equation}
is the distance between position vectors $\mathbf{r}_{i}(t)$ and $\mathbf{r}_{j}(t)$, and $c$ is the speed of light.

Equation~(\ref{eq:state}) is the \emph{true} state of a link. No node ever has direct access to it. Observations arrive after a long propagation delay, go missing during outages, and carry their own measurement noise. Every node's picture of the network is therefore partial and out of date by construction, not by some correctable flaw in the sensing.

\section{PROPOSED KNOWLEDGE-CENTRIC DIGITAL TWIN FRAMEWORK}
\label{sec:framework}
\subsection{OPERATIONAL KNOWLEDGE STATE: INTEGRATING COMMUNICATION PHYSICS, LEARNING, AND UNCERTAINTY}

As the true state is never directly available, autonomous decisions have to run on a continuously updated estimate rather than on the state itself. We call this estimate the \emph{Operational Knowledge State}. It combines a deterministic physics model with a learned residual and Bayesian uncertainty quantification. Its goal is not to synchronize exactly with the physical system because no digital twin (DT) can do so under these delays. Instead, it remains as informative as the available observations allow.

The true communication state introduced in Eq.~(\ref{eq:state}) can be decomposed into a deterministic component governed by communication physics and a residual component representing effects not captured by the analytical model. Accordingly, we can write

\begin{equation}
\mathbf{s}_{ij}(t)
=
f_{\mathrm{phys}}
\!\left(
R_{ij}(t),
\dot{R}_{ij}(t)
\right)
+
\boldsymbol{\epsilon}_{ij}(t),
\label{eq:decomposition}
\end{equation}
where $f_{\mathrm{phys}}(\cdot)$ denotes the deterministic communication model, $R_{ij}(t)$ is the instantaneous separation between nodes $i$ and $j$, $\dot{R}_{ij}(t)=dR_{ij}(t)/dt$ is the corresponding range rate, and $\boldsymbol{\epsilon}_{ij}(t)$ represents the residual communication dynamics not explained by the deterministic model, including propagation impairments, synchronization errors, hardware imperfections, and other stochastic effects. Throughout the manuscript, bold lowercase Latin and Greek symbols denote vectors, while bold uppercase Latin and Greek symbols denote matrices. Scalars are written in ordinary italic type.

The deterministic part follows from the propagated trajectories. Assuming free-space LoS propagation, received SNR follows the Friis equation in logarithmic form~\cite{friis_note_1946},

\begin{equation}
\mathrm{SNR}^{\mathrm{phys}}_{ij}(t)
=
P_{\mathrm{tx}}
+
G_{\mathrm{tx}}
+
G_{\mathrm{rx}}
-
20\log_{10}
\left(
\frac{4\pi R_{ij}(t)f_c}{c}
\right)
-
N_0,
\label{eq:friis}
\end{equation}
where $P_{\mathrm{tx}}$ is transmit power, $G_{\mathrm{tx}}$ and $G_{\mathrm{rx}}$ are the transmit and receive antenna gains, $f_c$ is carrier frequency, and $N_0$ is the noise floor. The second last factor in \eqref{eq:friis} is free-space path loss at the current range $R_{ij}(t)$.

Expected Doppler shift follows directly from the first-order Doppler approximation for relative radial motion, which can be written as

\begin{equation}
f^{D,\mathrm{phys}}_{ij}(t)
=
-
f_c
\frac{\dot{R}_{ij}(t)}{c},
\end{equation}
where $f_c$ is the carrier frequency. The negative sign follows the conventional Doppler sign convention, whereby a positive range rate (nodes moving apart) produces a negative frequency shift, while a negative range rate (nodes approaching one another) produces a positive frequency shift.
A neural network estimates that residual as follows

\begin{equation}
\widehat{\boldsymbol{\epsilon}}_{ij}(t)
=
r_{\theta}
\left(
h_{ij}(t)
\right),
\label{eq:residual}
\end{equation}
where $r_{\theta}(\cdot)$ is the residual network with weights $\theta$, and $h_{ij}(t)$ is the observation history for link $(i,j)$, defined as

\begin{equation}
h_{ij}(t)
=
\left\{
\mathbf{z}_{ij}(t_1),
\mathbf{z}_{ij}(t_2),
\ldots,
\mathbf{z}_{ij}(t_n)
\right\},
\end{equation}
with $\mathbf{z}_{ij}(t_k)$ the observation received at time $t_k\le t$, comprising every delayed SNR, Doppler, and delay measurement collected so far.

The network does not stop at a point estimate. It also predicts how much to trust it, which can be defined as

\begin{equation}
\boldsymbol{\Sigma}_{r,ij}(t)
=
g_{\theta}
\left(
h_{ij}(t)
\right),
\label{eq:residualcov}
\end{equation}
where $g_{\theta}(\cdot)$ is the network's uncertainty head. Small $\boldsymbol{\Sigma}_{r,ij}(t)$ means high confidence, while large $\boldsymbol{\Sigma}_{r,ij}(t)$ means the observations feeding it were sparse, noisy, or stale.

Physics and residual combine into the observation actually available at the receiver, which can be illustrated as

\begin{equation}
\mathbf{z}_{ij}(t)
=
f_{\mathrm{phys}}
\left(
R_{ij}(t),
\dot{R}_{ij}(t)
\right)
+
\widehat{\boldsymbol{\epsilon}}_{ij}(t),
\label{eq:observation}
\end{equation}
where $\mathbf{z}_{ij}(t)$, unlike the inaccessible true state $\mathbf{s}_{ij}(t)$, is what estimation actually has to work with. Its uncertainty $\boldsymbol{\Sigma}_{r,ij}(t)$ is fed directly into the Bayesian update below as the measurement-noise covariance.

Since $\mathbf{z}_{ij}(t)$ is delayed and imperfect, decisions have to run on a posterior over the true state given everything observed so far. We can define the observation history of link $(i,j)$ as

\begin{equation}
\mathcal{O}_{1:t}
=
\left\{
\mathbf{z}_{ij}(t_1),
\mathbf{z}_{ij}(t_2),
\ldots,
\mathbf{z}_{ij}(t_n)
\right\}
\end{equation}
so the posterior belief can be written as

\begin{equation}
b_{ij}(t)
=
p
\left(
\mathbf{s}_{ij}(t)
\mid
\mathcal{O}_{1:t}
\right),
\label{eq:belief}
\end{equation}
i.e., what the system currently knows about the link, not the link's actual physical state.

Following standard practice in Bayesian filtering~\cite{thrun2005probabilistic,barshalom}, we approximate this posterior as Gaussian, which can be written as

\begin{equation}
b_{ij}(t)
\approx
\mathcal{N}
\left(
\boldsymbol{\mu}_{ij}(t),
\boldsymbol{\Sigma}_{ij}(t)
\right),
\label{eq:gaussianbelief}
\end{equation}
with mean

\begin{equation}
\boldsymbol{\mu}_{ij}(t)
=
\left[
\widehat{\mathrm{SNR}}_{ij}(t),
\widehat{f}^{D}_{ij}(t),
\widehat{\tau}^{p}_{ij}(t)
\right]^{\mathrm T},
\label{eq:posteriormean}
\end{equation}
where $\widehat{\mathrm{SNR}}_{ij}(t)$, $\widehat{f}^{D}_{ij}(t)$, and $\widehat{\tau}^{p}_{ij}(t)$ denote the posterior estimates of the received SNR, Doppler frequency shift, and one-way propagation delay, respectively.

The corresponding posterior covariance can be illustrated as

\begin{equation}
\boldsymbol{\Sigma}_{ij}(t)
=
\mathbb{E}_{b_{ij}}
\left[
\left(
\mathbf{s}_{ij}(t)
-
\boldsymbol{\mu}_{ij}(t)
\right)
\left(
\mathbf{s}_{ij}(t)
-
\boldsymbol{\mu}_{ij}(t)
\right)^{\mathrm T}
\right],
\label{eq:posteriorcov}
\end{equation}
whose diagonal elements represent the estimation variances of the individual communication-state variables, while the off-diagonal elements capture their statistical correlations. Together, the pair $(\boldsymbol{\mu}_{ij}(t),\boldsymbol{\Sigma}_{ij}(t))$ provides the posterior belief associated with communication link $(i,j)$, where $\boldsymbol{\mu}_{ij}(t)$ represents the estimated communication-state vector, while $\boldsymbol{\Sigma}_{ij}(t)$ determines the confidence in that estimate through its associated uncertainty. Smaller covariance values correspond to higher confidence, whereas larger values indicate greater uncertainty.

This belief does not sit still between observations. Orbital motion keeps changing the environment even when nothing new has been measured, so the belief goes through a prediction step and then, whenever a fresh observation lands, it goes through the correction step.

Prediction propagates the mean through the physics model, which can be defined as

\begin{equation}
\boldsymbol{\mu}^{-}_{ij}(t)
=
f_{\mathrm{phys}}
\left(
\boldsymbol{\mu}_{ij}(t_{\mathrm{last}}),
t-t_{\mathrm{last}}
\right),
\label{eq:predictionmean}
\end{equation}
where $t_{\mathrm{last}}$ is when the last observation was fused. With no new evidence to anchor it, uncertainty grows as follows

\begin{equation}
\boldsymbol{\Sigma}^{-}_{ij}(t)
=
\boldsymbol{\Sigma}_{ij}(t_{\mathrm{last}})
+
\mathbf{Q}_{\mathrm{proc}}
\left(
t-t_{\mathrm{last}}
\right),
\label{eq:predictioncov}
\end{equation}
where $\mathbf{Q}_{\mathrm{proc}}$ is the process-noise covariance matrix accounting for orbital dynamics and everything else drifting in the absence of new data.

When a new observation arrives, its uncertainty
$\boldsymbol{\Sigma}_{r,ij}(t)=g_\theta(h_{ij}(t))$
(Eq.~\ref{eq:residualcov})
is treated as the measurement-noise covariance in a standard Bayesian
precision-weighted update~\cite{thrun2005probabilistic,barshalom},
which can be written as

\begin{equation}
\boldsymbol{\Sigma}^{-1}_{ij}(t)
=
\left(
\boldsymbol{\Sigma}^{-}_{ij}(t)
\right)^{-1}
+
\left(
\boldsymbol{\Sigma}_{r,ij}(t)
\right)^{-1},
\label{eq:covupdate}
\end{equation}

\begin{equation}
\boldsymbol{\mu}_{ij}(t)
=
\boldsymbol{\Sigma}_{ij}(t)
\left[
\left(
\boldsymbol{\Sigma}^{-}_{ij}(t)
\right)^{-1}
\boldsymbol{\mu}^{-}_{ij}(t)
+
\left(
\boldsymbol{\Sigma}_{r,ij}(t)
\right)^{-1}
\mathbf{z}_{ij}(t)
\right].
\label{eq:meanupdate}
\end{equation}
A confident observation pulls the belief toward it, whereas a noisy one barely moves it. This is what lets delayed, heterogeneous observations from different assets combine into one coherent estimate despite intermittent connectivity.

Collecting the updated link beliefs gives the probabilistic core of the
Operational Knowledge State. Unlike a conventional information or Kalman
filter, which typically maintains the posterior of a specified latent state,
the KCDT maintains a graph-indexed collection of link beliefs whose membership
changes with the time-varying communication topology. The Operational
Knowledge State is therefore defined as

\begin{equation}
\mathcal{K}(t)
=
\left(
G(t),
\left\{
b_{ij}(t)
\right\}_{(i,j)\in E(t)}
\right),
\label{eq:knowledge}
\end{equation}

where \(G(t)\) represents the communication relationships that are currently
available and
\[
b_{ij}(t)
=
p\!\left(
\mathbf{s}_{ij}(t)\mid\mathcal{D}_{ij}(0{:}t)
\right)
\]
denotes the posterior belief for active link \((i,j)\), conditioned on the
delayed and incomplete evidence
\(\mathcal{D}_{ij}(0{:}t)\) assimilated up to time \(t\). Under the Gaussian
approximation in Eq.~(\ref{eq:gaussianbelief}), each belief is represented by
its posterior mean \(\boldsymbol{\mu}_{ij}(t)\) and covariance \(\boldsymbol{\Sigma}_{ij}(t)\).

Equation~(\ref{eq:knowledge}) should not be interpreted as a single enlarged
Kalman-filter state or as the direct statistical fusion of all links into one
state vector. Fusion is performed within each link belief
\(b_{ij}(t)\) when multiple assets provide evidence concerning the same
communication state. The collection
\(\{b_{ij}(t)\}_{(i,j)\in E(t)}\) then organizes these link-level posterior
beliefs according to the current network graph. This separation allows the
KCDT to preserve link-specific uncertainty while maintaining a unified
network-wide representation for communication and mission reasoning.

The Operational Knowledge State is persistent even when observations are
unavailable. During an observation gap,
Eqs.~(\ref{eq:predictionmean}) through (\ref{eq:predictioncov}) propagate each
affected belief using the communication-physics model, learned residual, and
process uncertainty. Consequently, the posterior covariance continues to
evolve rather than remaining frozen at its last observed value. When a delayed
observation becomes available,
Eqs.~(\ref{eq:covupdate}) through (\ref{eq:meanupdate}) update the corresponding
belief according to the uncertainty of both the prediction and the
observation.

The posterior covariance contained in \(b_{ij}(t)\) determines the link's
Knowledge Entropy through Eq.~(\ref{eq:entropycorr}), whereas the assimilation
time of its supporting evidence determines Knowledge Freshness through
Eq.~(\ref{eq:freshness}). Mission-specific entropy and freshness requirements
are subsequently applied through Eq.~(\ref{eq:pmc}) and the admissible-link
definition in Eq.~(\ref{eq:admissiblelinks}). Thus, the KCDT connects four
elements that are generally treated separately: the time-varying network
topology, probabilistic state estimation, evidence timeliness, and
mission-specific knowledge requirements.

This characterization also clarifies the relationship with existing
frameworks. Information and Kalman filtering provide mechanisms for updating
the individual posterior beliefs. AoI provides the basis for measuring the age
of assimilated evidence. AoII evaluates realized incorrectness relative to the
underlying source state. VoI evaluates the expected decision benefit of a
prospective update. Semantic or task-oriented communication evaluates the
task relevance of communicated content. The KCDT does not replace these
frameworks. Instead, it maintains the persistent Operational Knowledge State
on which their corresponding estimation, scheduling, and decision mechanisms
may operate.

\subsection{CONSISTENCY PROPERTIES OF THE OPERATIONAL KNOWLEDGE STATE}
\label{sec:theory}

The previous subsection defined how the KCDT maintains each link belief. The following results state two consistency properties of the selected inference mechanism. They specialize established Gaussian fusion and covariance-propagation results to the KCDT and are included to make their assumptions and operational consequences explicit rather than to claim new Bayesian estimation identities.

\textbf{Relation to existing fusion and timeliness metrics.} Propositions~\ref{prop:fusion} and~\ref{prop:aging} instantiate standard Gaussian information-form fusion and covariance-propagation results~\cite{thrun2005probabilistic,barshalom}. They are not claimed as new standalone estimation theorems. Their role is to state the consistency conditions required by the KCDT inference layer. Informative evidence should reduce posterior uncertainty, whereas uncertainty should not decrease during an observation gap without new evidence. The contribution lies in embedding these established operations within a cislunar Operational Knowledge State that combines a time-varying communication graph, physics-informed prediction, a learned residual with predictive uncertainty, delayed observations, and mission-specific requirements.

Knowledge Freshness in Eq.~(\ref{eq:freshness}) is deliberately an AoI-derived clock that measures the time since the latest observation was assimilated into a link belief. It is not claimed as a mathematically new age process in isolation. Its role is to complement Knowledge Entropy because equally aged beliefs can have different posterior uncertainties. Knowledge Entropy is also distinct from AoII~\cite{maatouk_age_2020}. AoII penalizes the duration and severity of a mismatch between an estimate and the underlying source state, whereas entropy quantifies posterior uncertainty when the true state is not promptly observable. A biased but overconfident belief may therefore have low entropy despite being incorrect, making calibration an essential accompanying requirement. Finally, VoI quantifies the expected marginal benefit of acquiring or transmitting a particular update for a specified decision~\cite{wang_voi_2023}. The KCDT instead maintains the belief on which such a VoI calculation may operate, while mission completion measures whether mission-wide entropy and freshness requirements remain satisfied. Semantic or task-oriented communication assigns importance to message content relative to a task~\cite{gunduz2023beyond,qin2021semantic}. In contrast, the KCDT concerns the persistent probabilistic state assembled from heterogeneous evidence. These concepts are complementary rather than equivalent. Table~\ref{tab:frameworkcomparison} summarizes this comparison.

{\color{red}
\begin{table*}[t]
\centering
\caption{Relationship between the proposed KCDT and established estimation, timeliness, and task-value frameworks.}
\label{tab:frameworkcomparison}
\scriptsize
\renewcommand{\arraystretch}{1.15}
\setlength{\tabcolsep}{3pt}
\begin{tabular}{p{0.15\textwidth}p{0.22\textwidth}p{0.27\textwidth}p{0.28\textwidth}}
\toprule
\textbf{Framework} & \textbf{Primary object} & \textbf{Quantity represented} & \textbf{Relationship to the KCDT} \\
\midrule
Information/Kalman filtering & Latent dynamical state & Recursive posterior mean and covariance & Used as an inference mechanism. It does not by itself define the network-wide mission representation. \\
AoI & Latest delivered status update & Elapsed time since update generation or assimilation & Knowledge Freshness is an AoI-derived clock attached to a maintained link belief and interpreted jointly with posterior entropy. \\
AoII & Estimate relative to source truth & Duration/severity of incorrect estimation & Requires a correctness or error notion relative to the source. Knowledge Entropy instead represents posterior uncertainty when prompt ground-truth verification is unavailable. \\
VoI & Candidate observation/update & Expected marginal decision utility & Can prioritize observations using the KCDT belief. It is not a replacement for the persistent belief itself. \\
Semantic/task-oriented communication & Message/task representation & Meaning preservation or task performance & Operates on task-relevant content. The KCDT maintains the uncertain operational state on which tasks and decisions depend. \\
KCDT (this work) & Time-varying graph of link beliefs & Physics-informed posterior uncertainty, evidence age, and mission admissibility & Integrates estimation, timeliness, and mission requirements into one persistent cislunar operational representation. \\
\bottomrule
\end{tabular}
\end{table*}
}

\begin{proposition}[Knowledge Fusion Dominance]
\label{prop:fusion}
Consider link $(i,j)$ observed by $N$ independent communication assets, with posterior beliefs

\[
b_{ij}^{(k)}(t)
\sim
\mathcal N
\left(
\boldsymbol{\mu}_{ij}^{(k)}(t),
\boldsymbol{\Sigma}_{ij}^{(k)}(t)
\right),
\qquad
k=1,\ldots,N.
\]

Fusing these beliefs by Bayesian precision weighting gives

\begin{equation}
\boldsymbol{\Sigma}^{-1}_{ij,\mathrm{fusion}}(t)
=
\sum_{k=1}^{N}
\left(
\boldsymbol{\Sigma}_{ij}^{(k)}(t)
\right)^{-1},
\label{eq:fusioncov}
\end{equation}
and

\begin{equation}
\boldsymbol{\Sigma}_{ij,\mathrm{fusion}}(t)
\preceq
\boldsymbol{\Sigma}_{ij}^{(k)}(t),
\qquad
\forall k,
\label{eq:fusionineq}
\end{equation}
where $\preceq$ is positive-semidefinite ordering. Equality holds only when a single source carries all the information.
\end{proposition}

\begin{proof}
Each $\boldsymbol{\Sigma}_{ij}^{(k)}(t)$ is symmetric positive definite, so $(\boldsymbol{\Sigma}_{ij}^{(k)}(t))^{-1}\succ0$ for every $k$. Equation~(\ref{eq:fusioncov}) sums these positive-definite precision matrices, so

\[
\boldsymbol{\Sigma}^{-1}_{ij,\mathrm{fusion}}(t)
\succeq
\left(
\boldsymbol{\Sigma}_{ij}^{(k)}(t)
\right)^{-1},
\qquad
\forall k.
\]
Matrix inversion reverses the Loewner order on positive-definite matrices~\cite{horn13, bhatia}, so

\[
\boldsymbol{\Sigma}_{ij,\mathrm{fusion}}(t)
\preceq
\boldsymbol{\Sigma}_{ij}^{(k)}(t),
\qquad
\forall k,
\]
which is Bayesian fusion never adding uncertainty. When more than one independent observation contributes, the inequality is strict, i.e., the fused belief beats every individual one.
\end{proof}

This makes it clear that as more independent assets contribute observations, the fused belief is never less informative than any one of them alone. The DT accumulates knowledge and does not just average whatever comes in. This conclusion applies to the stated independent Gaussian case. Eqs.~(\ref{eq:correlatedfusioncov}) through (\ref{eq:cimean}) below give the appropriate alternatives when estimates are correlated.

\begin{corollary}[Entropy Reduction through Knowledge Fusion]
The differential entropy of a Gaussian belief can be written as

\begin{equation}
H_{ij}(t)
=
\frac{1}{2}
\log
\left(
(2\pi e)^d
\left|
\boldsymbol{\Sigma}_{ij}(t)
\right|
\right),
\label{eq:entropycorr}
\end{equation}
with $d$ the dimension of the state. Equation~(\ref{eq:fusionineq}) then gives

\begin{equation}
H_{ij,\mathrm{fusion}}(t)
\le
H_{ij}^{(k)}(t),
\qquad
\forall k.
\end{equation}
Fusion never increases entropy, and strictly reduces it whenever an independent, informative observation is added.
\end{corollary}

\begin{proposition}[Knowledge Aging]
\label{prop:aging}
Suppose no observations arrive over $[t_0,t]$. Under the prediction model of Eqs.~(\ref{eq:predictionmean}) through (\ref{eq:predictioncov}), the covariance evolves as

\begin{equation}
\boldsymbol{\Sigma}^{-}_{ij}(t)
=
\boldsymbol{\Sigma}_{ij}(t_0)
+
\mathbf{Q}_{\mathrm{proc}}
(t-t_0),
\label{eq:aging}
\end{equation}
with $\mathbf{Q}_{\mathrm{proc}}\succeq\mathbf{0}$. Then entropy in \eqref{eq:entropycorr} does not decrease.
\end{proposition}

\begin{proof}
Since $\mathbf{Q}_{\mathrm{proc}}\succeq\mathbf{0}$, Eq.~(\ref{eq:aging}) gives $\boldsymbol{\Sigma}^{-}_{ij}(t)\succeq\boldsymbol{\Sigma}_{ij}(t_0)$ for $t\ge t_0$. Determinants are monotone under the Loewner order on positive-definite matrices~\cite{horn13,bhatia}, so $|\boldsymbol{\Sigma}^{-}_{ij}(t)|\ge|\boldsymbol{\Sigma}_{ij}(t_0)|$. Since Gaussian differential entropy increases monotonically with the covariance determinant,

\[
H_{ij}(t)
\ge
H_{ij}(t_0),
\qquad
t\ge t_0,
\]
so knowledge entropy cannot fall without new observations.
\end{proof}

Propositions~1 and~2 pull in opposite directions, and both matter. Fusion pulls uncertainty down whenever new independent evidence arrives, aging pushes it back up whenever nothing new comes in. Together they describe how operational knowledge actually behaves under realistic, intermittent cislunar connectivity.

One consequence is worth stating on its own, there are scenarios where no single asset carries enough information to succeed alone, yet fusing several does.

\begin{corollary}[Fusion Necessity]
Suppose a task requires $\boldsymbol{\Sigma}_{ij}\preceq\boldsymbol{\Sigma}_{\mathrm{req}}$ for some maximum tolerable uncertainty $\boldsymbol{\Sigma}_{\mathrm{req}}$, and suppose every individual asset instead has

\[
\boldsymbol{\Sigma}_{ij}^{(k)}
\succ
\boldsymbol{\Sigma}_{\mathrm{req}},
\qquad
k=1,\ldots,N,
\]
so none of them alone is good enough. By Proposition~1, $\boldsymbol{\Sigma}_{ij,\mathrm{fusion}}\preceq\boldsymbol{\Sigma}_{ij}^{(k)}$ for every $k$, so

\[
\boldsymbol{\Sigma}_{ij,\mathrm{fusion}}
\preceq
\boldsymbol{\Sigma}_{\mathrm{req}}
\]
can still hold once enough independent observations are fused. Distributed fusion can succeed where every individual asset fails on its own. This is the theoretical basis for the Fusion Necessity scenario in Section~\ref{sec:results}: each asset there knows only part of the picture and fails the mission requirement alone, but the fused Operational Knowledge State clears it.
\end{corollary}

\textbf{Correlated observations.} Proposition~\ref{prop:fusion} assumes that the contributing estimates are conditionally independent given the true link state. This assumption is an idealization for the Longjiang-2 examples in Section~\ref{sec:longjiang}. The direct and lunar-reflected paths share transmitter oscillator and clock errors, while ground stations observing the same downlink share satellite effective isotropic radiated power (EIRP) and orbital-prediction uncertainties. Related integrated-network studies also illustrate how common network states and learned models couple otherwise distributed nodes~\cite{TWC2026correlated,JSAC2025correlated}. To represent correlation explicitly, define the stacked estimates and repeated-state observation matrix as
\begin{equation}
\widehat{\mathbf{s}}_{ij}(t)
=
\begin{bmatrix}
\boldsymbol{\mu}_{ij}^{(1)}(t)^{\mathrm T} & \cdots & \boldsymbol{\mu}_{ij}^{(N)}(t)^{\mathrm T}
\end{bmatrix}^{\mathrm T},
\qquad
\mathbf{A}=\mathbf{1}_{N}\otimes \mathbf{I}_d .
\label{eq:stackedestimates}
\end{equation}
Let $\widetilde{\mathbf{s}}_{ij}^{(k)}(t)=\boldsymbol{\mu}_{ij}^{(k)}(t)-\mathbf{s}_{ij}(t)$ denote the estimation-error vector and let $\boldsymbol{\Sigma}_{\mathrm{joint}}(t)$ contain the diagonal blocks $\boldsymbol{\Sigma}_{ij}^{(k)}(t)$ and cross-covariance blocks
\begin{equation}
\boldsymbol{\Sigma}_{ij}^{(k,\ell)}(t)
=
\mathbb{E}\!\left[
\widetilde{\mathbf{s}}_{ij}^{(k)}(t)
\widetilde{\mathbf{s}}_{ij}^{(\ell)}(t)^{\mathrm T}
\right],
\qquad k\neq\ell .
\label{eq:crosscovariance}
\end{equation}
For jointly Gaussian unbiased estimates with known nonsingular joint covariance, generalized least-squares fusion gives
\begin{align}
\boldsymbol{\Sigma}_{ij,\mathrm{fusion}}(t)
&=
\left(
\mathbf{A}^{\mathrm T}\boldsymbol{\Sigma}_{\mathrm{joint}}^{-1}(t)\mathbf{A}
\right)^{-1},
\label{eq:correlatedfusioncov}\\
\boldsymbol{\mu}_{ij,\mathrm{fusion}}(t)
&=
\boldsymbol{\Sigma}_{ij,\mathrm{fusion}}(t)
\mathbf{A}^{\mathrm T}\boldsymbol{\Sigma}_{\mathrm{joint}}^{-1}(t)
\widehat{\mathbf{s}}_{ij}(t).
\label{eq:correlatedfusionmean}
\end{align}
Equation~(\ref{eq:fusioncov}) is recovered when all cross-covariance blocks are zero. With the complete joint covariance, optimal generalized fusion is no worse than selecting any single unbiased source because single-source selection remains a feasible estimator. Correlation can nevertheless reduce, or for fully redundant information eliminate, the strict gain obtained by adding another source. Applying Eq.~(\ref{eq:fusioncov}) while incorrectly setting nonzero cross-covariances to zero produces an overconfident covariance and exaggerates the associated entropy reduction.

When cross-covariances cannot be estimated reliably, covariance intersection (CI) provides a conservative alternative:
\begin{align}
\boldsymbol{\Sigma}_{ij,\mathrm{CI}}^{-1}(t)
&=
\sum_{k=1}^{N}\omega_k(t)
\left(\boldsymbol{\Sigma}_{ij}^{(k)}(t)\right)^{-1},
\label{eq:cicovariance}\\
\boldsymbol{\mu}_{ij,\mathrm{CI}}(t)
&=
\boldsymbol{\Sigma}_{ij,\mathrm{CI}}(t)
\sum_{k=1}^{N}\omega_k(t)
\left(\boldsymbol{\Sigma}_{ij}^{(k)}(t)\right)^{-1}
\boldsymbol{\mu}_{ij}^{(k)}(t),
\label{eq:cimean}
\end{align}
where $\omega_k(t)\geq0$ and $\sum_k\omega_k(t)=1$. Equations~(\ref{eq:cicovariance}) and (\ref{eq:cimean}) are the covariance and
mean of the normalized weighted geometric mean
$\prod_{k}[\mathcal{N}(\mathbf{s}_{ij};\boldsymbol{\mu}_{ij}^{(k)},\boldsymbol{\Sigma}_{ij}^{(k)})]^{\omega_k}$
of the contributing Gaussian beliefs~\cite{julier1997nondivergent}, obtained by
completing the square in the resulting quadratic form. The convexity constraint
$\sum_k\omega_k(t)=1$ replaces the precision summation of
Eq.~(\ref{eq:fusioncov}) by a convex combination of precisions, so shared
information is never counted more than once and
$\boldsymbol{\Sigma}_{ij,\mathrm{CI}}(t)\succeq\boldsymbol{\Sigma}_{ij,\mathrm{fusion}}(t)$
for any weights. Setting $\omega_m(t)=1$ recovers source $m$ alone, so
single-source selection remains a feasible vertex of the weight simplex. The weights may be selected to minimize $\operatorname{tr}(\boldsymbol{\Sigma}_{ij,\mathrm{CI}})$ or $\log\det(\boldsymbol{\Sigma}_{ij,\mathrm{CI}})$. The empirical root mean square error (RMSE) improvements reported for the Longjiang-2 examples remain evidence that combining sources can improve point estimation, but their confidence must be interpreted through the joint covariance or a conservative fusion rule rather than the independence formula.

\textbf{Cross-link dependence of hybrid-model errors.}
The preceding formulation addresses correlated estimates of the same link.
A different dependence can arise among the prediction errors of distinct
links. Let
\(\mathcal{L}(t)=\{\ell_1,\ldots,\ell_L\}\), where each
\(\ell=(i,j)\in E(t)\) denotes an active communication link. To characterize this dependence mathematically, we decompose the hybrid-model
error for link \(\ell\) into components attributable to the physical model and the learned residual as follows
\begin{equation}
\mathbf{e}_{\ell}(t)
=
\mathbf{e}_{\ell}^{\mathrm{phys}}(t)
+
\mathbf{e}_{\ell}^{\mathrm{res}}(t).
\label{eq:hybridmodelerror}
\end{equation}
This decomposition does not assume that the two components are independent.
For two active links \(\ell\) and \(m\), their joint model-error covariance
block is
\begin{align}
\mathbf{P}_{\ell m}^{e}(t)
&=
\mathbb{E}
\left[
\mathbf{e}_{\ell}(t)\mathbf{e}_{m}^{\mathrm T}(t)
\right]
\nonumber\\
&=
\mathbf{P}_{\ell m}^{\mathrm{pp}}(t)
+
\mathbf{P}_{\ell m}^{\mathrm{pr}}(t)
+
\mathbf{P}_{\ell m}^{\mathrm{rp}}(t)
+
\mathbf{P}_{\ell m}^{\mathrm{rr}}(t),
\label{eq:crosslinkmodelcov}
\end{align}
where
\begin{equation}
\mathbf{P}_{\ell m}^{ab}(t)
=
\mathbb{E}
\left[
\mathbf{e}_{\ell}^{a}(t)
\left(\mathbf{e}_{m}^{b}(t)\right)^{\mathrm T}
\right],
\qquad
a,b\in\{\mathrm{phys},\mathrm{res}\}.
\label{eq:crosslinkcomponents}
\end{equation}
For \(\ell=m\), the terms
\(\mathbf{P}_{\ell\ell}^{\mathrm{pr}}\) and
\(\mathbf{P}_{\ell\ell}^{\mathrm{rp}}\) describe dependence between the physical-model
and learned-residual errors within the same link. For \(\ell\neq m\), the four
terms capture cross-link dependence. Such off-diagonal blocks can arise when
links share spacecraft ephemerides, clocks, transmitter parameters,
environmental conditions, propagation-model approximations, neural-network
parameters, or overlapping observation histories.

The evaluated implementation maintains the beliefs \(b_{ij}(t)\) separately
and represents prediction and update uncertainty through the per-link
quantities \(\mathbf{Q}_{\mathrm{proc}}\) in Eq.~(\ref{eq:predictioncov}) and
\(\boldsymbol{\Sigma}_{r,ij}(t)\) in Eq.~(\ref{eq:residualcov}). It therefore does not
explicitly propagate the complete block covariance
\begin{equation}
\mathbf{P}^{e}(t)
=
\left[
\mathbf{P}_{\ell m}^{e}(t)
\right]_{\ell,m\in\mathcal{L}(t)}.
\label{eq:networkmodelcov}
\end{equation}
Neglecting the within-link cross-terms can miscalibrate an individual link
belief, while neglecting the off-diagonal cross-link blocks can misrepresent
the joint uncertainty of decisions involving several links, including
multi-hop routing, handover, and coupled resource allocation. Consequently,
the present results establish the behavior of the evaluated factorized
implementation and do not establish calibration of a fully joint
correlation-aware network belief.

\subsection{QUANTIFYING OPERATIONAL KNOWLEDGE}

Conventional link metrics such as BER, throughput, latency, and packet-delivery ratio characterize transmission performance, while AoI characterizes the timeliness of status updates. None of these quantities alone jointly represents the posterior uncertainty and evidence age of a multisource operational belief. In an autonomous Earth to Moon mission, communication supports navigation and coordination, but it is not the goal itself. Decisions must be judged on the quality of the knowledge behind them rather than only on the quality of the signal.

Consequently, we introduce a small set of knowledge-centric metrics as follows.

\subsubsection{Knowledge Entropy}

Knowledge Entropy is the uncertainty in the Operational Knowledge State itself, not in the channel. Since each link's belief is Gaussian, its uncertainty is just the differential entropy of that Gaussian.

For link $(i,j)$, the differential entropy of a Gaussian belief $H_{ij}(t)$ is defined in \eqref{eq:entropycorr}. The determinant captures both per-variable uncertainty and correlation between SNR, Doppler, and delay in one number. A lower $H_{ij}(t)$ indicates better knowledge. A higher value indicates that delays, outages, or sparse measurements have increased uncertainty in the estimate. Proposition~2 shows that $H_{ij}(t)$ rises during outages. Proposition~1 shows that fusion reduces it.

At mission level, entropy aggregates across all active links as follows

\begin{equation}
H(t)
=
\sum_{(i,j)\in E(t)}
w_{ij}
H_{ij}(t),
\label{eq:networkentropy}
\end{equation}
with weights $w_{ij}$, $\sum_{(i,j)\in E(t)} w_{ij}=1$, set by mission importance. A crewed-operation link or a safety-critical control channel gets more weight than a delay-tolerant telemetry feed, so $H(t)$ reflects mission relevance, not just raw uncertainty.

\subsubsection{Knowledge Freshness}

Low uncertainty by itself is not enough, an estimate that was accurate ten seconds ago may already be wrong in a fast-changing orbit. Knowledge Freshness, $\Delta_{ij}(t)$, is the time since the belief was last updated by a real observation, which can be defined as

\begin{equation}
\Delta_{ij}(t)
=
t
-
t_{\mathrm{last}}(i,j).
\label{eq:freshness}
\end{equation}
This is an AoI-derived clock attached to the maintained link belief. Its contribution is not a new age process in isolation, but its joint interpretation with Knowledge Entropy, i.e, two beliefs with identical freshness may have different uncertainty because they are supported by observations of different quality.

\subsubsection{Probability of Mission Completion}

A cislunar communication system ultimately exists to get a mission done, not to maximize throughput. So we raise a question: across the whole mission, how often do all mission-critical links stay within their required entropy and freshness bounds?

Let $\mathcal{M}_{\mathrm{crit}}$ be the mission-critical links. Then, Probability of Mission Completion (PMC) can be defined as

\begin{equation}
\begin{aligned}
\mathrm{PMC} = \Pr\Big[\, &\forall (i,j)\in\mathcal{M}_{\mathrm{crit}},\ \forall t\in[0,T]: \\
&H_{ij}(t)\le H^{\mathrm{req}}_{ij} \ \land\ \Delta_{ij}(t)\le \Delta^{\mathrm{req}}_{ij} \,\Big],
\end{aligned}
\label{eq:pmc}
\end{equation}
with $H^{\mathrm{req}}_{ij}$ and $\Delta^{\mathrm{req}}_{ij}$ the mission's entropy and freshness requirements for link $(i,j)$. We estimate PMC directly from Monte Carlo runs, as the fraction of mission realizations that stay within these bounds for the full horizon.

\subsection{KNOWLEDGE-TO-DECISION INTERFACE}
\label{sec:decisioninterface}

The Operational Knowledge State becomes actionable through an admissibility-and-risk interface. At time $t$, define the set of links whose uncertainty and evidence age satisfy the mission requirements as
\begin{equation}
\mathcal{E}_{\mathrm{adm}}(t)
=
\left\{
(i,j)\in E(t):
H_{ij}(t)\leq H_{ij}^{\mathrm{req}},\;
\Delta_{ij}(t)\leq\Delta_{ij}^{\mathrm{req}}
\right\}.
\label{eq:admissiblelinks}
\end{equation}
For every admissible link, define the normalized knowledge-aware cost
\begin{equation}
c_{ij}(t)
=
\alpha_{\tau}\widehat{\tau}^{p}_{ij}(t)
+\alpha_H\frac{H_{ij}(t)}{H_{ij}^{\mathrm{req}}}
+\alpha_{\Delta}\frac{\Delta_{ij}(t)}{\Delta_{ij}^{\mathrm{req}}}
-\alpha_S\widehat{\mathrm{SNR}}_{ij}(t),
\label{eq:knowledgecost}
\end{equation}
where the nonnegative coefficients express mission priorities. These definitions yield concrete reference policies for the decision categories in Fig.~\ref{fig3b}.

\emph{Adaptive routing:} among feasible paths $\mathcal{P}_{\mathrm{adm}}(t)$ between the source and destination whose constituent links belong to $\mathcal{E}_{\mathrm{adm}}(t)$, select
\begin{equation}
p^{\star}(t)
=
\arg\min_{p\in\mathcal{P}_{\mathrm{adm}}(t)}
\sum_{(i,j)\in p}c_{ij}(t).
\label{eq:routingpolicy}
\end{equation}
If no admissible path exists, transmission is deferred or routed through a mission-defined safe fallback rather than acting on an overconfident link estimate.

\emph{Proactive link switching:} let $(i_c,j_c)$ be the current link and $(i_a,j_a)$ the lowest-cost admissible alternative. A switch is triggered when the current link leaves $\mathcal{E}_{\mathrm{adm}}(t)$ or when
\begin{equation}
c_{i_a j_a}(t)+\eta<c_{i_c j_c}(t),
\label{eq:switchpolicy}
\end{equation}
for a hysteresis margin $\eta>0$, preventing oscillatory switching between statistically similar links.

\emph{Uncertainty-aware resource allocation:} for a resource budget $B(t)$, allocate nonnegative resources $x_{ij}(t)$ by
\begin{equation}
\begin{aligned}
\mathbf{x}^{\star}(t)
=\arg\min_{\mathbf{x}\succeq\mathbf{0}}\quad
&\sum_{(i,j)\in E(t)}w_{ij}
H^{+}_{ij}\!\left(t,x_{ij}\right)\\
\text{s.t.}\quad
&\sum_{(i,j)\in E(t)}x_{ij}(t)\leq B(t),
\end{aligned}
\label{eq:resourcepolicy}
\end{equation}
where $H^{+}_{ij}(t,x_{ij})$ is the predicted post-allocation entropy obtained from the covariance update under the observation quality enabled by $x_{ij}$. Higher-priority or more uncertain links therefore receive resources when doing so produces the greatest mission-weighted reduction in posterior uncertainty. Equations~(\ref{eq:admissiblelinks}) through (\ref{eq:resourcepolicy}) specify reproducible reference policies that consume the KCDT state. Mission-specific optimal control and joint policy benchmarking remain separate extensions.

\subsection{BASELINE METHODS, EVALUATION PROTOCOL, AND COMPUTATIONAL COMPLEXITY}
\label{sec:baselines}

This subsection lays out the communication hops used, the baselines themselves, and the complexity argument behind the scalability result.

\subsubsection{Regime B communication hops}

Regime~A represents Earth-orbiting relays under two-body propagation. Regime~B represents the Lunar Gateway using its NRHO approximation. Within Regime~B, we simulated three hops, each corresponding to a different segment of the Gateway communication chain. Hop~1 is the Gateway-to-relay crosslink, which is occasionally blocked as the Gateway moves through its orbit. Hop~2 is the direct Gateway-to-ground downlink, which is visible for almost the entire orbit and is interrupted only briefly near perilune. Hop~3 is the Gateway-to-lunar-surface link and has the longest and least predictable outages because the lunar geometry, rather than Earth visibility, blocks the view.

\subsubsection{Baseline communication strategies}

We compare four alternatives to the full KCDT, plus a zero-delay oracle used only as an upper bound, rather than a competitor.

\emph{State-centric.} It treats the last observation as still valid, with no growth in uncertainty between observations, the naive assumption tested in Fig.~\ref{fig1}.

\emph{Physics-only.} It uses $f_{\mathrm{phys}}(\cdot)$ from Eq.~(\ref{eq:friis}) directly, no learned residual, fixed uncertainty. This isolates what the learned residual and its uncertainty actually add on top of orbital geometry and link-budget physics alone.

\emph{Data-driven.} A standard Kalman-style filter with no physics prior, i.e., the prediction step in Eq.~(\ref{eq:predictionmean}) is replaced by a generic random walk. This isolates the physics prior's contribution from the opposite direction.

\emph{Analytic residual heuristic.} This is used only for the Regime B/hop2 comparison, as a closed-form, exponentially decaying residual covariance, $\boldsymbol{\Sigma}_r(\tau)=\boldsymbol{\Sigma}_0\exp(-\tau/\tau_c)$, standing in for the learned network of Eq.~(\ref{eq:residual}). Its correlation time $\tau_c$ was fit once on Regime A and reused unchanged on hop2, deliberately, to see what happens when a fusion method's assumptions are not re-tuned for a new regime. This is a different object from the neural residual model $r_\theta(\cdot)$ tested across regimes elsewhere in the Results.

\emph{Fixed schedule.} It transmits on a set cadence regardless of link condition and represents conventional pre-planned mission communication. It is used only in Fig.~\ref{fig6}. It does not produce a state estimate, so RMSE and calibration do not apply. It provides a qualitative reference for mission outcomes rather than estimation accuracy.

\emph{Oracle.} This is a zero delay, zero uncertainty, observes the true state $\mathbf{s}_{ij}(t)$ directly every timestep. It is included in Table~\ref{table2} as a ceiling no delay-constrained method can beat.

Every cell of the mission-level estimator comparison uses $N=20$ independent episodes for each method and regime combination. The same visibility and outage pattern is held fixed across methods within a cell, so differences in RMSE, calibration, and PMC arise from the estimator rather than from different scenario realizations.

\subsubsection{Computational complexity and scalability}

Each link's belief $b_{ij}(t)$ (Eq.~\ref{eq:belief}) is maintained separately, so prediction and update through Eqs.~(\ref{eq:predictionmean}) to (\ref{eq:meanupdate}) operate only on a small and fixed $d\times d$ covariance for a single link, with $d=3$ throughout. Under the factorized implementation evaluated in this paper, a full fusion step over the network costs can be written as

\begin{equation}
\mathcal{O}\!\left(|E(t)|\,d^3\right),
\label{eq:computational_complexity}
\end{equation}
which is linear in the number of active links $|E(t)|$ for fixed $d$, not quadratic or cubic, as a fully joint filter over the whole network would be. Correlation-aware observation fusion replaces each independent-source
update by Eq.~(\ref{eq:correlatedfusioncov}) within a correlated source group.
Maintaining the cross-link model-error blocks in
Eq.~(\ref{eq:networkmodelcov}) would additionally require a joint or grouped
covariance over the affected links. The associated cost depends on the size
and structure of each correlated group and is not covered by the reported
linear benchmark. The factorized per-link implementation is what keeps the evaluated framework usable at cislunar scale, where active links can range from a few mission-critical channels to hundreds of telemetry and relay connections.

\begin{figure*}[t]
\centering
\subfloat[]{\includegraphics[width=0.4\textwidth]{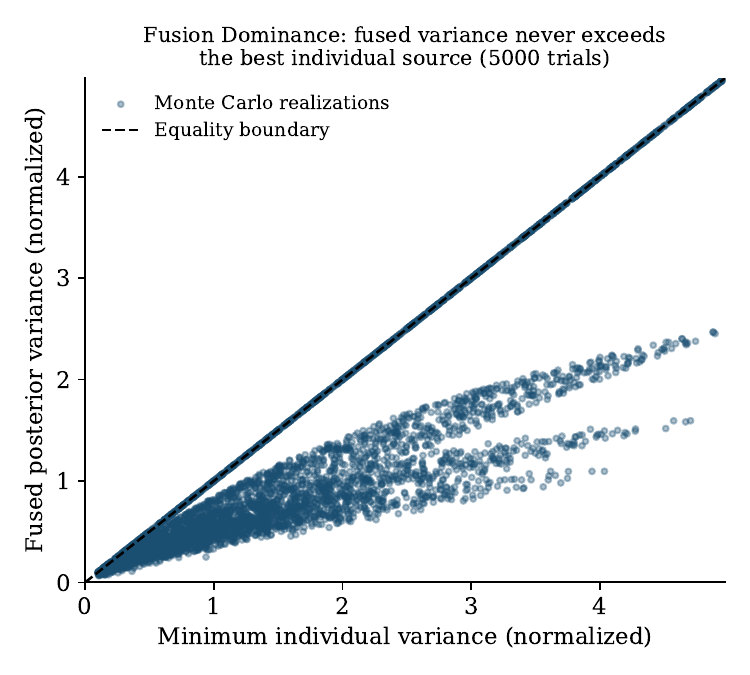}\label{fig5a}}
\hfil
\subfloat[]{\includegraphics[width=0.4\textwidth]{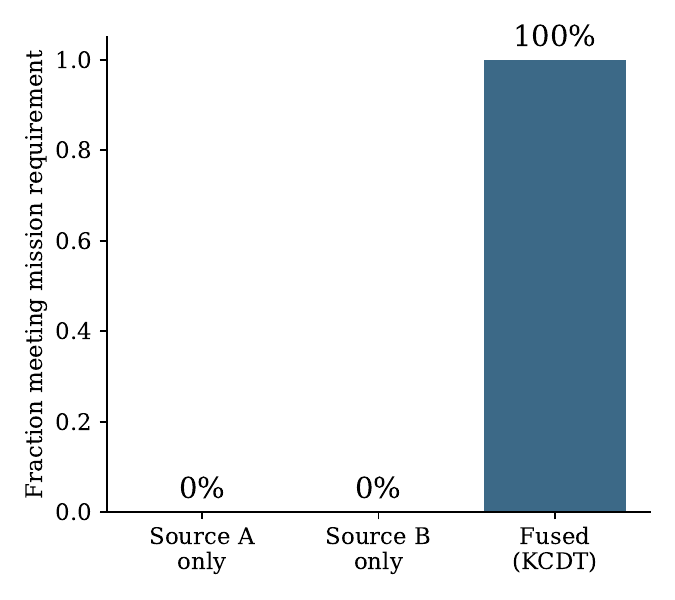} \label{fig5b}}
\hfil
\subfloat[]{\includegraphics[width=0.4\textwidth]{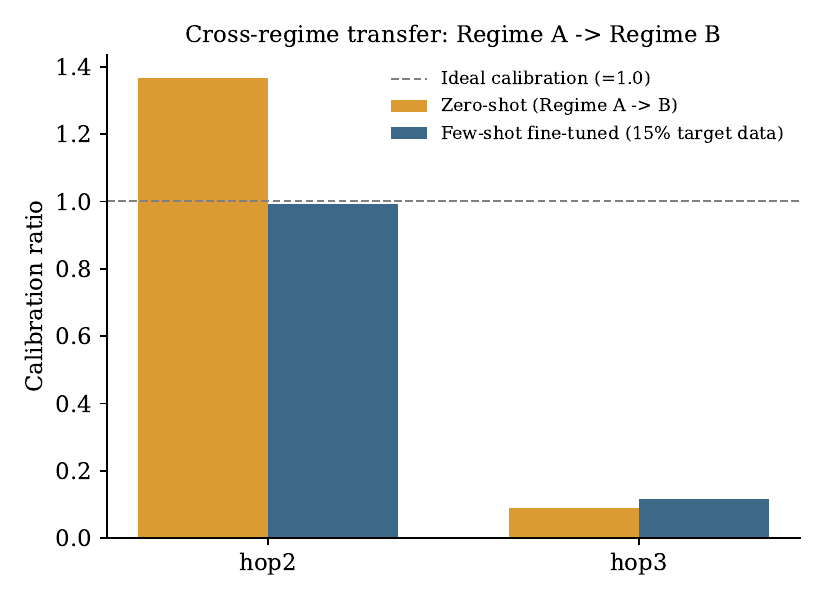} \label{fig5c}}
\caption{Distributed knowledge fusion.
(a) Fusion Dominance across 5000 randomized trials.
Each circle represents one Monte Carlo realization, with
the horizontal axis showing the minimum variance among
the individual observation sources and the vertical axis
showing the corresponding fused posterior variance.
(b) Fusion Necessity across 200 Monte
Carlo realizations. Neither Source A nor Source B
independently satisfies the mission's entropy requirement
(0\% success), whereas their fusion satisfies the
requirement in all realizations (100\% success).
(c) Cross-regime transfer of the learned residual model
from Regime A to Regime B. Calibration ratio
(ideal = 1.0, dashed line) for Hop 2 and Hop 3 in the
zero-shot setting compared with the ratio after few-shot
fine-tuning on 15\% of target-hop data.}
\label{fig5}
\end{figure*}

\section{EXPERIMENT RESULTS}
\label{sec:results}

We evaluated the KCDT through five complementary analyses. We first tested the fusion properties directly and examined whether combining sources is sometimes the only way to meet a mission requirement. We then examined how operational knowledge behaves during outages, tested cross-regime transfer of the residual model, and benchmarked the knowledge-state estimator against state-centric, physics-only, and data-driven alternatives. Finally, because the evaluation is otherwise simulation-based, we performed a physically grounded comparison using the reported orbit, operating envelope, anomaly, and very high frequency (VHF) and ultra high frequency (UHF) communication record of Longjiang-2~\cite{wei_design_2020}. This analysis tests whether the same qualitative behavior appears under externally reported mission conditions. It is not presented as an exact replication because the complete flight dataset and hardware tables were unavailable. Unless stated otherwise, the results are Monte Carlo simulations under the environments and protocol described in Section~\ref{sec:baselines}.

\subsection{DISTRIBUTED KNOWLEDGE CONSISTENTLY OUTPERFORMS INDIVIDUAL OBSERVATIONS}
\label{sec:fusionresults}

Distributed observations produced more informative Operational Knowledge than any individual communication source. Across 5000 randomized communication realizations, the posterior uncertainty after Bayesian knowledge fusion was lower than the uncertainty of the best individual observation source in every single trial, as illustrated in Fig.~\ref{fig5a}, a direct numerical check of the standard independent-Gaussian fusion property stated as Proposition~\ref{prop:fusion}. On average, fusion cut posterior variance by 31.5\% relative to the best individual source, which is the practical signature of a twin that keeps accumulating information rather than one that just averages whatever comes in.

That improvement turns out to matter operationally, not just statistically. To ask whether fusion is merely useful or actually necessary, we built a deliberately demanding Fusion Necessity scenario in which two observation sources were each individually informative, but neither was sufficient on its own to meet the mission's entropy requirement. As expected, Sources~A and B both failed the requirement in all 200 Monte Carlo realizations (0\% success). Once fused, the same requirement was met in every realization, namely, 100\% success as shown in Fig.~\ref{fig5b}.

These results demonstrate that some regimes cannot achieve mission success by improving only one observation source. In such cases, a properly fused Operational Knowledge State built from multiple heterogeneous sources is necessary. The principal benefit of the KCDT is therefore not limited to reducing estimation uncertainty. It can support decisions that no individual source can support alone.

\begin{table*}[t]
\centering
\caption{Cross-regime transfer of the learned residual model: RMSE and calibration ratio (ideal value 1.0) for the unbounded and bounded predictive-variance architectures, zero-shot from Regime~A to each Regime~B hop, and after few-shot fine-tuning on 15\% of target-hop data.}
\label{table1}
\begin{tabular}{lccc}
\toprule
\textbf{Evaluation} & \textbf{RMSE (dB)} & \textbf{Calibration (unbounded)} & \textbf{Calibration (bounded)} \\
\midrule
Regime A (held-out) & 0.316 & 1.027 & 1.050 \\
hop1 (zero-shot)    & 0.317 & 1.051 & 1.070 \\
hop2 (zero-shot)    & 0.346 & 1.347 & 1.403 \\
hop3 (zero-shot)    & 0.312 & 0.055 & 1.046 \\
hop2 (few-shot, bounded)   & not applicable & not applicable & 0.993 \\
hop3 (few-shot, bounded)   & not applicable & not applicable & 0.980 \\
\bottomrule
\end{tabular}
\end{table*}

\subsection{OPERATIONAL KNOWLEDGE TRANSFERS ACROSS COMMUNICATION REGIMES}
\label{sec:transferresults}

Operational Knowledge remained transferable across communication regimes despite substantial differences in observation patterns and communication dynamics. The residual model was trained exclusively using communication data from Regime~A and subsequently evaluated without retraining on the three communication hops of Regime~B (Table~\ref{table1}). This zero-shot evaluation assesses whether the learned knowledge representation captures general communication behavior rather than overfitting to a specific orbital configuration.

\begin{table}[t]
\centering
\caption{Mission-level estimator comparison under nominal link conditions for Regime~A and Regime~B/Hop~2 with $N=20$ episodes per cell. PMC is reported for each method. Oracle values are shown as a zero-delay upper bound rather than as a competing method.}
\label{table2}
\begin{tabular}{llccc}
\toprule
\textbf{Regime} & \textbf{Method} & \textbf{RMSE (dB)} & \textbf{Calibration} & \textbf{PMC} \\
\midrule
\multirow{5}{*}{A}
 & State-centric   & 1.513 & 2.371 & 0.406 \\
 & Physics-only    & 0.712 & 0.127 & 0.000 \\
 & Data-driven     & 1.511 & 0.107 & 0.398 \\
 & KCDT (proposed) & 0.706 & 0.652 & 0.406 \\
 & Oracle          & 0.000 & 0.000 & 0.406 \\
\midrule
\multirow{5}{*}{B / hop2}
 & State-centric   & 0.397 & 0.176  & 0.987 \\
 & Physics-only    & 2.488 & 1.558  & 0.000 \\
 & Data-driven     & 0.529 & 0.662  & 0.979 \\
 & KCDT (proposed) & 2.464 & 29.531 & 0.983 \\
 & Oracle          & 0.000 & 0.000  & 0.987 \\
\bottomrule
\end{tabular}
\end{table}

The learned model generalized to Hop~1 without retraining and achieved an RMSE of 0.317~dB and a calibration ratio of 1.051. This value is sufficiently close to the ideal value of 1.0 to indicate that the residual dynamics remained well calibrated under conditions that the model did not encounter during training. Hop~3 was the outlier because the calibration ratio of the original unbounded predictive-variance architecture decreased to 0.055. Examination of the traces showed that several outages on Hop~3 lasted as long as 79.7~h, which exceeded the longest training gap of 17.6~h. This result indicates extrapolation failure rather than a fundamental failure of the residual model.

The correction imposed a cap on the predictive variance at the stationary variance of the residual process. Replacing the unbounded variance head with this sigmoid-bounded version increased the Hop~3 calibration ratio to 1.046, while RMSE and calibration in Regime~A and Hop~1 remained essentially unchanged (Table~\ref{table1}). Thus, maintaining a physically plausible uncertainty estimate makes Operational Knowledge more trustworthy during outages that are substantially longer than those represented in the training data.

Hop~2 required additional adaptation. Even after bounding, its calibration ratio remained at 1.39, which indicates that its residual statistics differ from those of Regime~A. Fine-tuning on 15\% of the Hop~2 data reduced the ratio to 0.993. The same procedure improved Hop~3 from 0.46 to 0.980 (Fig.~\ref{fig5c}). Both values are within 2\% of the ideal calibration ratio. These results show that a small amount of target-domain data can be effective when the underlying representation already transfers.

\begin{figure*}[t]
\centering
\subfloat[]{\includegraphics[width=1\textwidth]{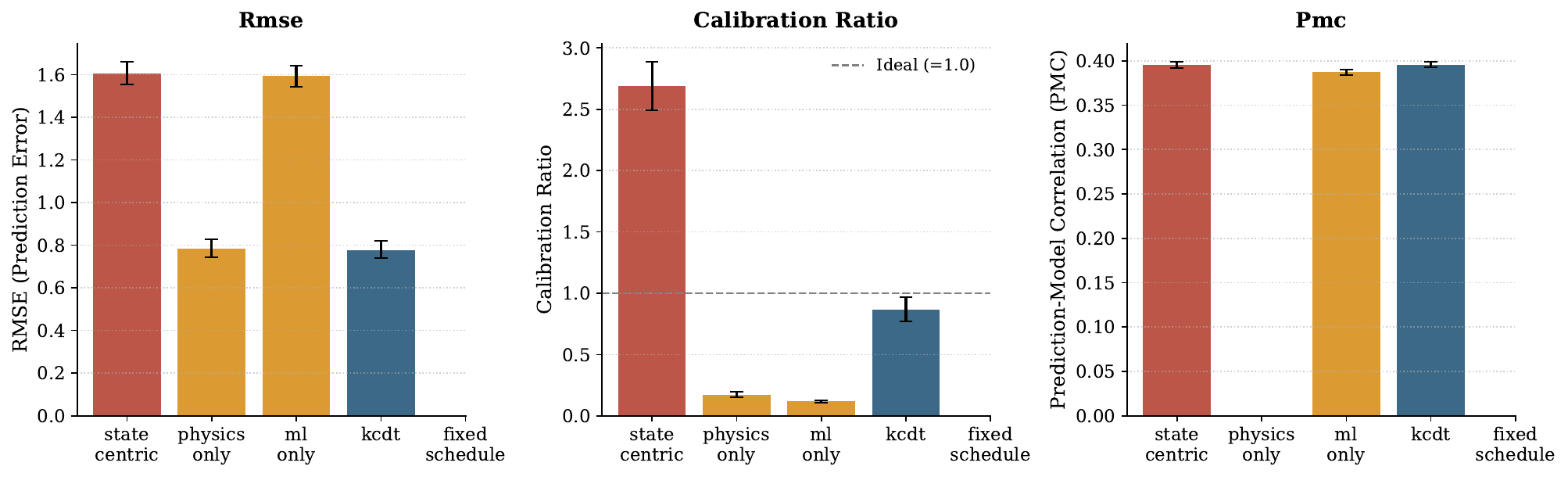}}
\hfil
\subfloat[]{\includegraphics[width=1\textwidth]{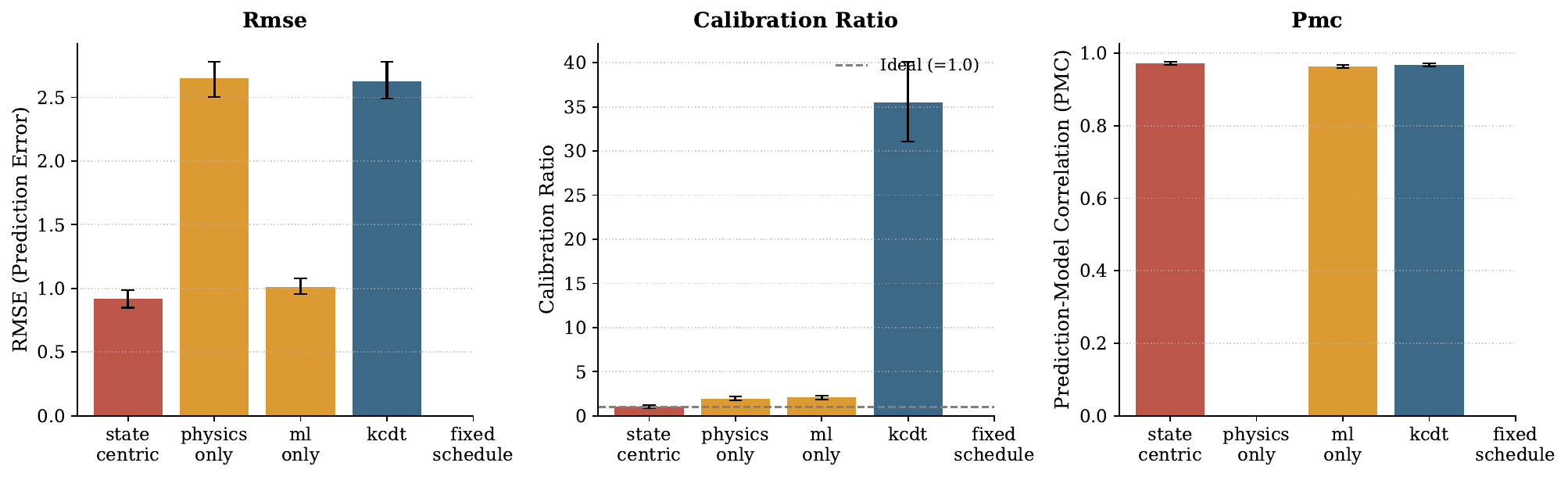}}
\caption{\textbf{Mission-level comparison across methods.} State-centric, physics-only, data-driven, KCDT, and fixed-schedule alternatives under nominal conditions in (a)~Regime~A and (b)~Regime~B/Hop~2. The first four curves compare estimators under the common mission test. The fixed schedule is a qualitative communication-policy reference that transmits without using the knowledge state and is not assigned RMSE or calibration values.}
\label{fig6}
\end{figure*}

Taken together, this points to a residual model that has learned something genuinely transferable about communication behavior, not something tied to one orbital configuration. Most of what it learned carried over directly, the rest was recoverable with a modest amount of new data. That matters for cislunar missions specifically, where you rarely get much representative training data up front and where conditions can shift meaningfully between missions.

\subsection{MISSION-LEVEL PERFORMANCE AGAINST BASELINE METHODS}
\label{sec:missionevaluation}

\begin{figure*}[t]
\centering
\subfloat[]
{\includegraphics[width=0.45\textwidth]{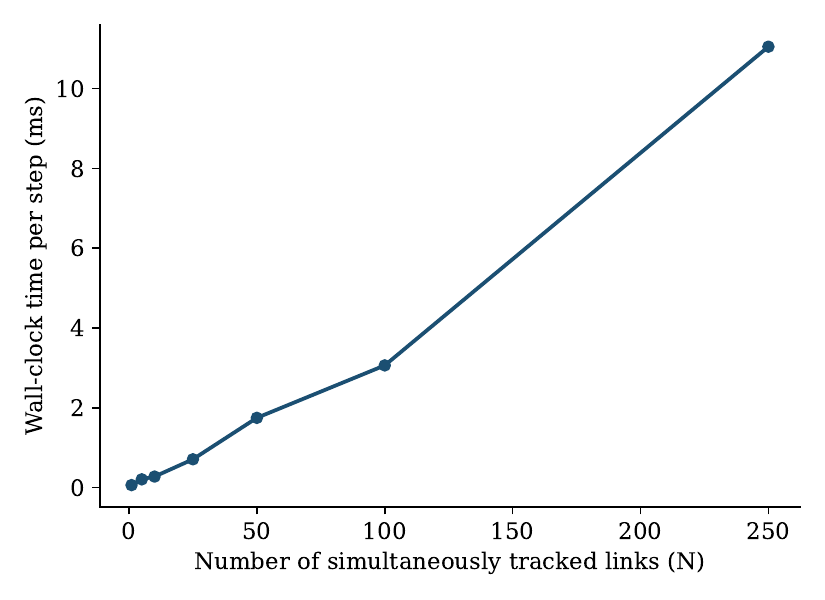}\label{fig7a}}
\hfil
\subfloat[]
{\includegraphics[width=0.8\textwidth]{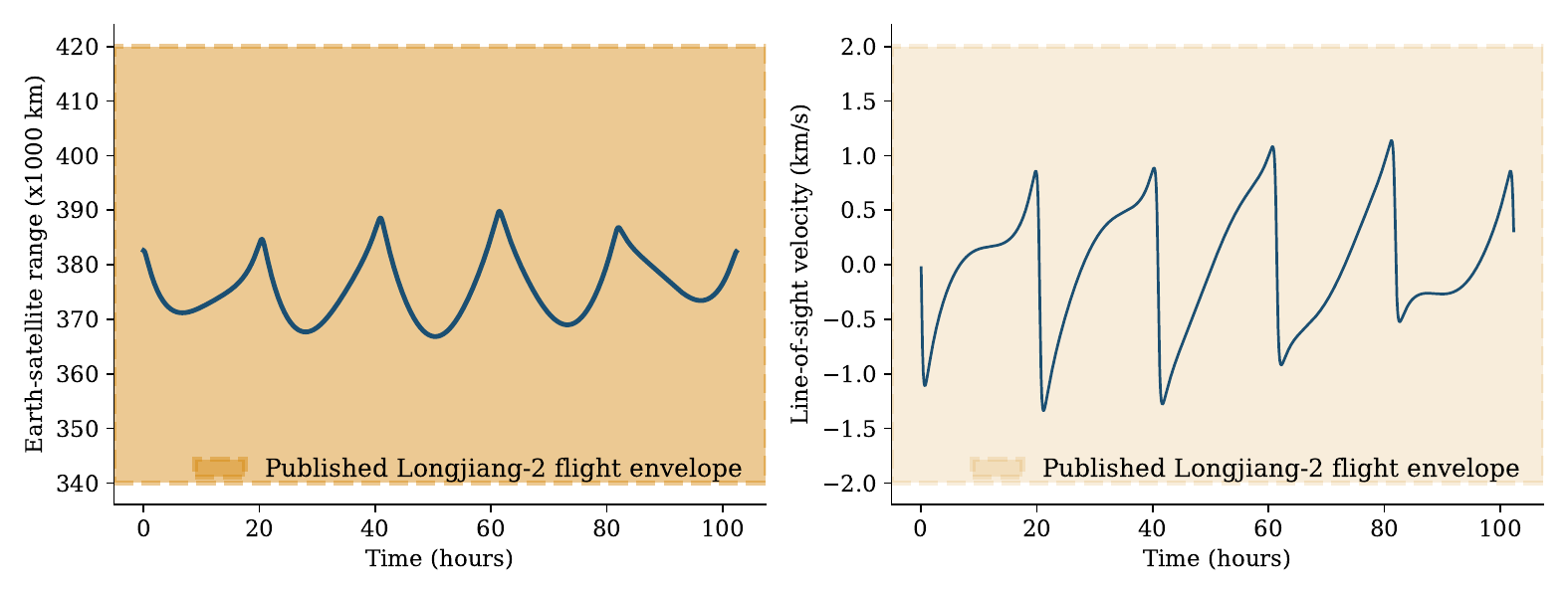}\label{fig7b}}
\caption{\textbf{Scalability of knowledge fusion per time step.} (a)~Wall-clock time per fusion step as a function of the number of simultaneously tracked links $N$. \textbf{Orbit fidelity against the Longjiang-2 flight envelope.} (b)~Simulated range from Earth to the satellite over time using the achieved Longjiang-2 lunar orbit compared with the reported flight range (shaded), and simulated line-of-sight velocity compared with the reported bound (shaded).}
\label{fig7}
\end{figure*}

We benchmarked the KCDT against three simpler baselines, namely, a state-centric heuristic, a physics-only predictor, and a data-driven (Kalman-style) filter with no physics prior. In addition, we took a zero-delay oracle as an upper-bound reference under nominal link conditions in Regime~A and Regime~B/Hop~2 (Table~\ref{table2}). In Regime~A, the KCDT achieved the lowest RMSE of any non-oracle method (0.706~dB) and a calibration ratio of 0.652, compared with 2.371 for the state-centric baseline. This result indicates that the state-centric method is less accurate and overconfident in its estimated accuracy. PMC was 0.406 for both the KCDT and the state-centric baseline under the entropy and freshness requirements used here. It was 0.000 for the physics-only method because its fixed, nonadaptive uncertainty estimate never satisfied the freshness requirement, regardless of the accuracy of its point predictions.

That pattern did not carry over to Regime~B/Hop~2. The correlation-time parameter of the analytic residual heuristic used for this comparison, as defined in Section~\ref{sec:results}, was fitted to Regime~A and transferred unchanged to Hop~2. This mismatch increased its calibration ratio to 29.5, which exceeded the ratios of the state-centric (0.176) and data-driven (0.662) baselines. PMC remained high across every method on Hop~2, with values from 0.979 to 0.987, because visibility is nearly continuous under nominal conditions. Consequently, there is little opportunity for any method to fail. The extent to which a fusion approach outperforms simpler baselines therefore depends strongly on whether its residual-model parameters were fitted to the deployment regime. This finding is consistent with the cross-regime transfer results above.

Figure~\ref{fig6} summarizes the mission-level comparison and additionally includes a fixed, pre-planned transmission schedule that ignores the Operational Knowledge State. The fixed schedule gives the weakest qualitative mission outcome in both regimes. This observation should not be interpreted as an optimization benchmark because the experiment does not implement the routing, switching, or resource-allocation policies of Section~\ref{sec:decisioninterface}. The comparatively higher KCDT RMSE in Regime~B/Hop~2 is caused by the deliberately mismatched analytic residual heuristic, whose correlation-time parameter was fitted in Regime~A and transferred unchanged. Thus, the experiment demonstrates sensitivity to uncertainty-model calibration rather than universal superiority of knowledge-centric estimation. The explicit policies in Eqs.~(\ref{eq:routingpolicy}) through (\ref{eq:resourcepolicy}) provide reproducible decision rules for future end-to-end policy evaluation.

We also checked that maintaining the knowledge state actually scales the way the factorized per-link implementation predicts (Section~\ref{sec:baselines}). Wall-clock time per fusion step went from about 0.03~ms at $N=1$ tracked link to about 6~ms at $N=250$ links (Fig.~\ref{fig7a}), tracking the linear $O(|E(t)|)$ complexity rather than degrading super-linearly as the network grows. Since the state-centric baseline does not perform Bayesian uncertainty propagation or distributed knowledge fusion, its computational complexity is inherently lower. The purpose of Fig.~\ref{fig7a} is therefore not to compare computational cost across fundamentally different algorithms, but to verify linear scaling for the evaluated factorized KCDT implementation. As noted in Section~\ref{sec:baselines}, this measurement does not establish the complexity of joint correlation-aware fusion.

\subsection{FINDINGS HOLD UP AGAINST REAL FLIGHT DATA, NOT SIMULATION ALONE}
\label{sec:longjiang}

Everything above was evaluated in our simulated cislunar environment. We then examined whether the same qualitative properties appear under conditions reported for a real spacecraft. We assessed fusion, entropy behavior, transfer, and decision quality using \emph{Longjiang-2}, a lunar-orbit microsatellite with a flown VHF and UHF radio and published flight measurements~\cite{wei_design_2020}. Whereas the earlier checks used the nominal distance from Earth to the Moon, this comparison uses the reported operational envelope of a flown spacecraft.

\begin{figure*}[t]
\centering
\subfloat[]{\includegraphics[width=0.6\textwidth]{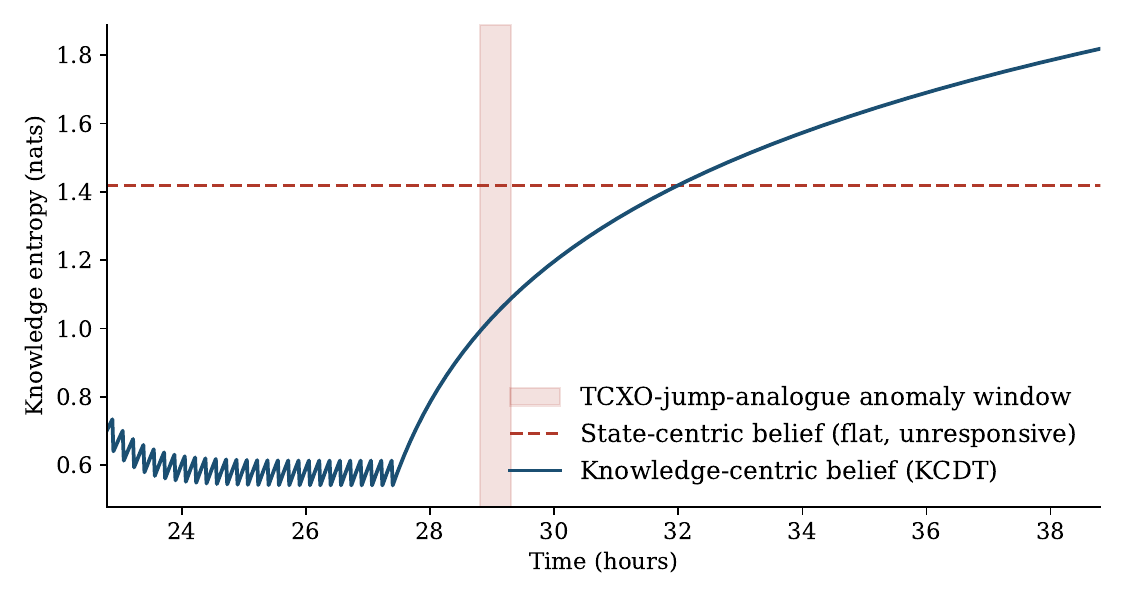}\label{fig8a}}
\hfil
\subfloat[]{\includegraphics[width=0.6\textwidth]{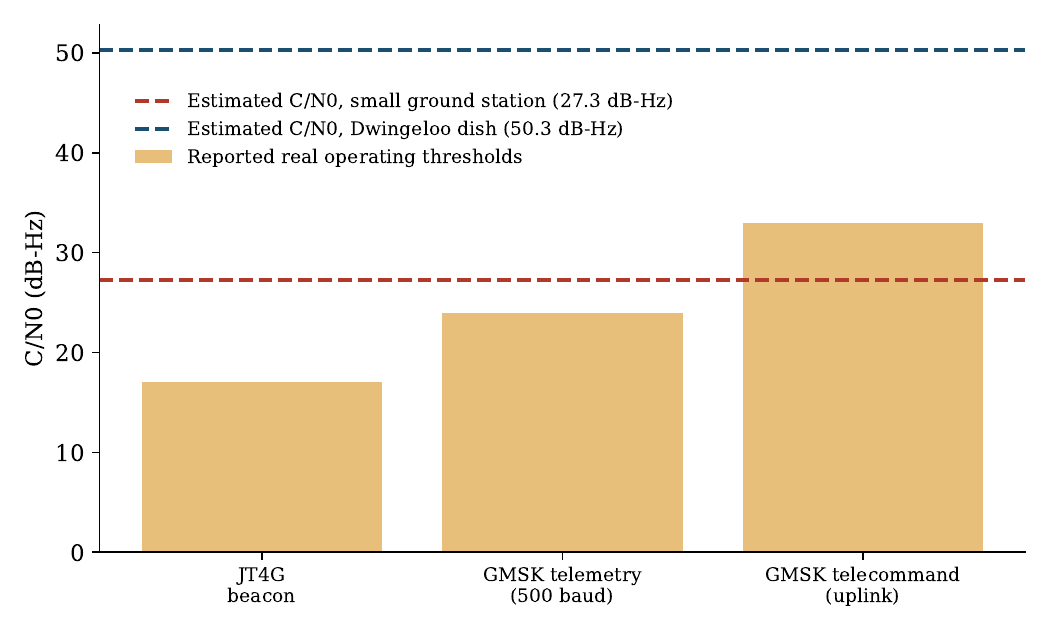}\label{fig8b}}
\caption{\textbf{Response to a real anomaly and link-budget plausibility.} (a)~Knowledge entropy for the knowledge-centric belief and the state-centric baseline around an injected discontinuity modelled on Longjiang-2's documented oscillator-jump anomaly. (b)~Estimated link margin for two ground-station classes against the paper's reported real operating thresholds.}
\label{fig8}
\end{figure*}

The first check examined geometric fidelity. Propagating the achieved Longjiang-2 lunar orbit with an altitude from 357~km to 13{,}704~km and combining it with the motion of the Moon around Earth produced a simulated range from Earth to the satellite of 366{,}830 to 389{,}800~km. It also produced a LoS velocity from $-1.34$ to $1.14$~km/s. Both results fall within the reported flight envelope of 340{,}000 to 420{,}000~km and approximately $\pm2$~km/s (Fig.~\ref{fig7b}). The second check reused this orbit to examine a documented anomaly. An in-flight jump of approximately 20~Hz in the temperature-compensated oscillator of the satellite reportedly corrupted downlink packets by causing the receiver phase-locked loop to lose lock. Injecting an equivalent unmodeled discontinuity into a simulated link left the entropy of the state-centric belief unchanged throughout the anomaly, which reproduced the structural blindness established earlier. In contrast, the entropy of the knowledge-centric belief increased by 0.176~nats during the anomaly and remained elevated briefly afterward (Fig.~\ref{fig8a}). This behavior is consistent with the report that the jump corrupted some, but not all, of the affected packets.

Fig.~\ref{fig8b} examines whether the link-budget model, using physically reasonable parameters for a small CubeSat at the simulated Longjiang-2 range, produces values near the reported operating thresholds. The exact transmit power and antenna gains are contained in supplementary tables that were unavailable to us. This is therefore a plausibility check rather than a replication, and the hardware parameters are explicitly identified as estimates. A small ground station produces an estimated link margin of 27.3~dB-Hz, which lies within the reported range of 17 to 33~dB-Hz for the Joe Taylor 4-tone G submode (JT4G) beacon, Gaussian minimum shift keying (GMSK) telemetry, and GMSK telecommand links. A Dwingeloo-sized dish produces 50.3~dB-Hz, which exceeds every reported threshold and is consistent with Dwingeloo providing the strongest signal among the four ground stations. These findings do not demonstrate that the estimated hardware parameters match the actual radio. They show only that the physical link-budget model produces plausible output at an externally reported operating point.

\begin{figure*}[t]
\centering
\subfloat[]{\includegraphics[width=0.9\textwidth]{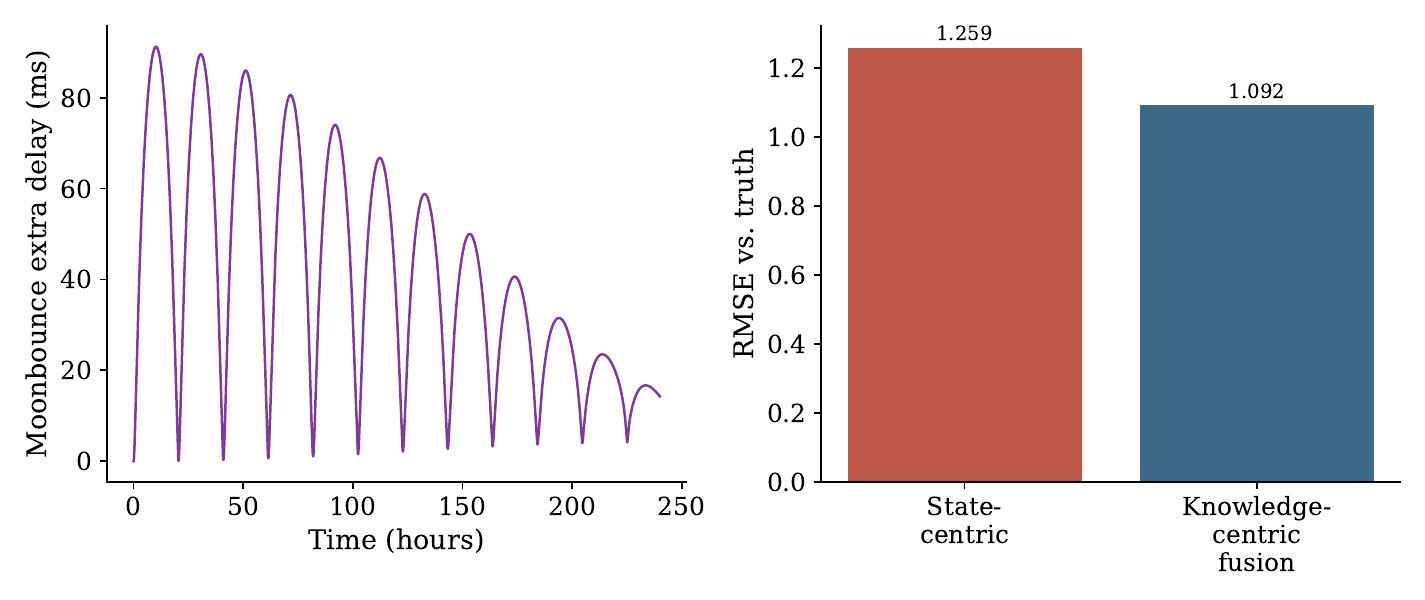}\label{fig9a}}
\hfil
\subfloat[]{\includegraphics[width=0.9\textwidth]{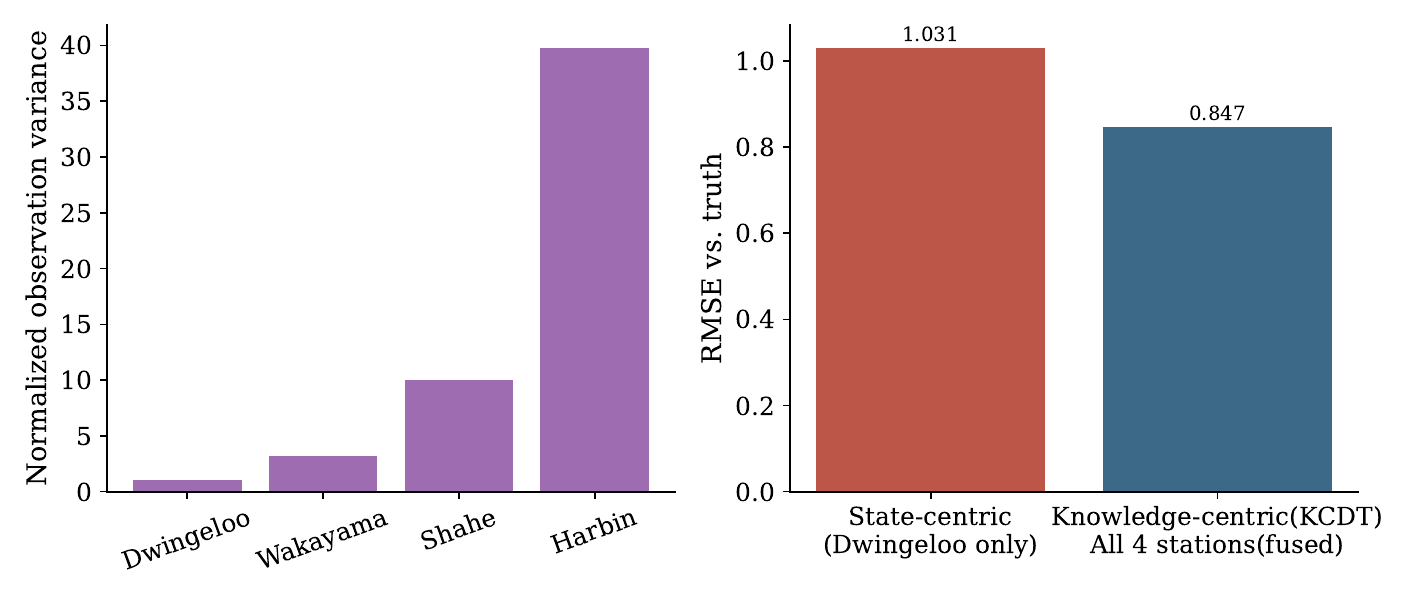}\label{fig9b}}
\caption{Two physically motivated fusion scenarios based
on Longjiang-2.
(a) Fusion of direct and lunar-reflected
signal paths, illustrating the additional propagation delay
and the RMSE improvement obtained through two-path fusion.
(b) Fusion of observations from four ground stations
with unequal signal quality, showing their normalized
observation variances and the RMSE improvement relative
to the strongest individual station (Dwingeloo).
The reported RMSE comparison uses the factorized
implementation. Shared-error effects and uncertainty
calibration require the correlation-aware formulation
in Eqs.~(31) and (32).}
\label{fig9}
\end{figure*}

The Longjiang-2 record also motivates two multisource reception cases, but neither case is strictly independent. The direct and lunar-reflected paths contain distinct propagation information while sharing transmitter oscillator, clock, EIRP, and orbital-prediction errors. Under the factorized implementation used for the reported point estimates, combining both paths reduced RMSE from 1.259 to 1.092 relative to the direct path alone (Fig.~\ref{fig9a}). Likewise, Dwingeloo, Wakayama, Shahe, and Harbin receive different local noise and scintillation realizations but share spacecraft EIRP, orbit and ephemeris, and transmitter-model errors. Combining all four observations reduced RMSE from 1.031 for Dwingeloo alone to 0.847 (Fig.~\ref{fig9b}). These reductions demonstrate complementary point-estimation information rather than calibrated covariance dominance under independence. A deployable confidence estimate should use the joint cross-covariance model in Eq.~(\ref{eq:crosscovariance}) or covariance intersection when those cross-covariances are unknown. Finally, using physically reasonable parameters for a small CubeSat in the link-budget model at the simulated Longjiang-2 range produced estimated link margins of 27.3~dB-Hz for a small ground station and 50.3~dB-Hz for a Dwingeloo-sized dish. Both values are consistent with the reported operating thresholds from 17 to 33~dB-Hz across the mission telemetry, beacon, and command links.

\begin{figure*}[t]
\centering
\includegraphics[width=0.98\textwidth]{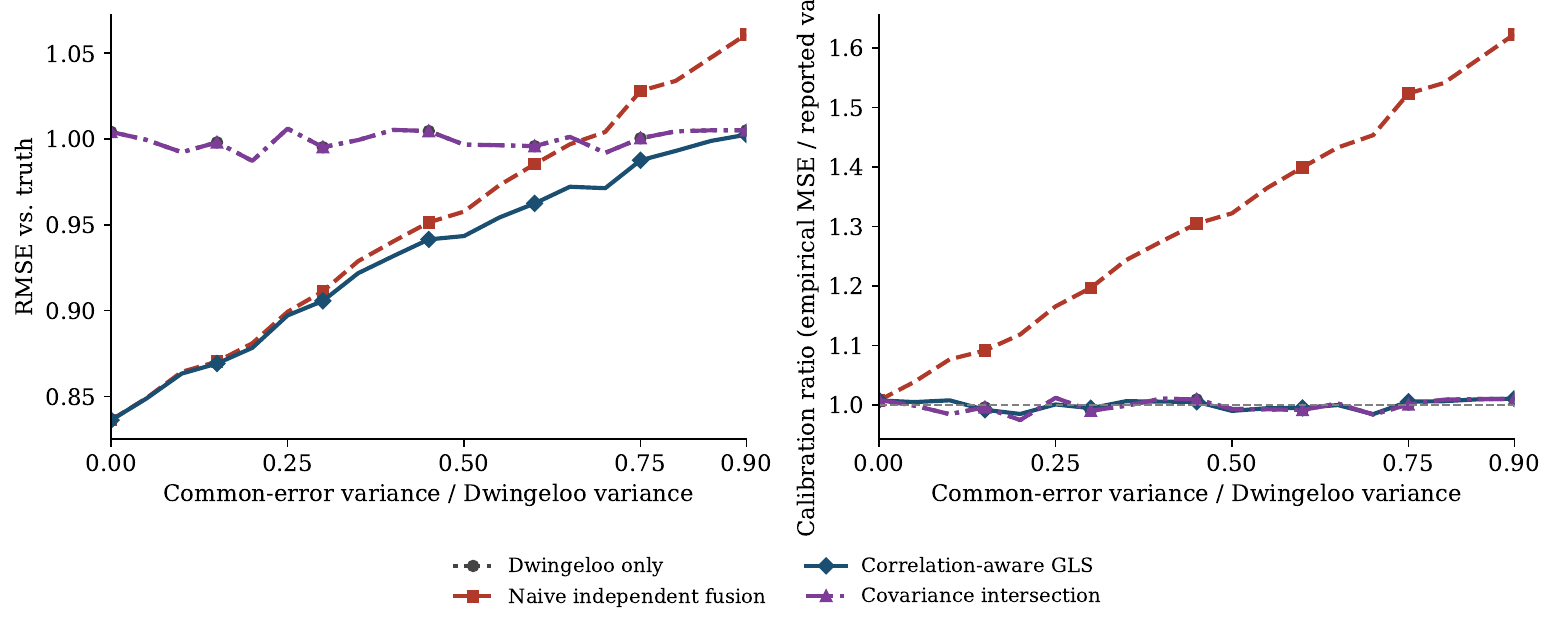}
\caption{\textbf{Sensitivity of four-station fusion to shared observation errors.} The common-error variance represents shared spacecraft EIRP, orbit/ephemeris, transmitter-model, and large-scale propagation uncertainties and is normalized by the marginal variance of the Dwingeloo observation.}
\label{fig10}
\end{figure*}

To quantify the consequence of the shared errors identified above, we supplemented the factorized four-station result with the correlation-sensitivity experiment in Fig.~\ref{fig10}. Each station error was decomposed into a common component and a station-specific component, giving the joint covariance $\boldsymbol{\Sigma}_{\mathrm{joint}}=\operatorname{diag}(\sigma_{1,\mathrm{local}}^2,\ldots,\sigma_{4,\mathrm{local}}^2)+\sigma_c^2\mathbf{1}\mathbf{1}^{\mathrm T}$. This covariance is a structured instance of Eq.~(\ref{eq:crosscovariance}). We held the four marginal station variances fixed and varied the common-error variance $\sigma_c^2$ from zero to $0.9$ times the Dwingeloo marginal variance, using 20{,}000 Monte Carlo trials at each point. With no common error, naive independent fusion and correlation-aware generalized least squares (GLS) coincide. Both attain an RMSE of 0.836 compared with 1.004 for Dwingeloo alone. As the common component grows, the naive method continues to report the independence-based covariance and becomes progressively overconfident. At the largest tested fraction, its RMSE rises to 1.061, exceeding the Dwingeloo-only RMSE of 1.005, while its calibration ratio increases to 1.623. In contrast, GLS uses the complete joint covariance and retains an RMSE of 1.003 and a calibration ratio of 1.011. Optimized scalar covariance intersection coincides with Dwingeloo alone because the variance-minimizing scalar solution selects the observation with the greatest precision when no cross-covariances are trusted. These results show that the point-estimation improvement in Fig.~\ref{fig9b} does not by itself justify an independence assumption. Correlation can reduce or eliminate the gain from additional stations, and reliable uncertainty reporting requires either the complete joint covariance or a conservative fusion rule.

Figure~\ref{fig10} further shows that the apparent fusion benefit is conditional on correctly representing shared errors and can disappear when correlation is strong. The marginal variance of every station is held fixed. Naive precision fusion becomes increasingly overconfident as the shared-error component grows, whereas GLS remains calibrated by using the complete joint covariance. In this scalar example, optimized covariance intersection selects the lowest-variance Dwingeloo observation and therefore coincides with the Dwingeloo-only curve. The assumed common-error fractions form a sensitivity analysis and are not estimates obtained from the Longjiang-2 measurements.

\section{DISCUSSION}
\label{sec:discussion}

Across the four analyses reported here and the fifth check against a flown spacecraft, a consistent pattern emerges. Representing what is known and the associated confidence provides information that a point estimate cannot. In the Fusion Necessity scenario, fusion determined whether the mission requirement could be satisfied. Entropy tracked realistic outages in a way that a fixed synchronization assumption cannot. It increased during observation gaps and decreased, but not to zero, when new evidence arrived. The learned residual model transferred across orbital regimes after limited fine-tuning when its predictive uncertainty was physically bounded. This correction resolved the extrapolation failure associated with the longest outages. The advantage of the decision layer over simpler baselines was also conditional rather than automatic. It was evident in Regime~A, became marginal under the nearly continuous visibility of Hop~2, and reversed when the residual-model assumptions were intentionally mismatched to the target regime. This dependence on the accuracy of the uncertainty model is important because knowledge-centric fusion is not universally superior. It provides the greatest benefit when its uncertainty estimates remain well calibrated, whereas inaccurate uncertainty modeling can degrade performance. Identifying and adapting to such model mismatch is therefore essential for reliable deployment.

Some limitations follow directly from how the evaluation was built, and the Gateway's orbital model is the one worth walking through rather than just flagging. As stated in Section~\ref{sec:results}, Regime~B's Lunar Gateway is propagated with a reduced-order approximation of its NRHO, not a full CR3BP integration. That choice has a real advantage, as it is fast, and speed matters when a result has to come from thousands of Monte Carlo episodes rather than a single flight trajectory, and it still reproduces the orbital period and the general shape of Gateway visibility windows well enough to drive realistic intermittent connectivity. However, it has a real cost attached to it. A true NRHO carries its own unstable manifold structure, needs periodic station-keeping burns to stay on it, and produces occultation timing a reduced-order stand-in can only approximate. So the exact timestamps, exact outage durations, and exact periodicity reported here should be read as representative of the \emph{kind} of intermittency a Gateway-class orbit produces, not as a forecast for any specific mission's actual schedule.

However, this simplification is sufficient to demonstrate the significance of the proposed framework at this stage. The theoretical properties established in this work, namely that distributed knowledge fusion reduces uncertainty while uncertainty increases during prolonged periods without new observations, hold for any covariance sequence satisfying the stated assumptions, regardless of where that sequence originates. Consequently, the entropy, freshness, and calibration results presented in Section~\ref{sec:results} depend primarily on the statistical characteristics of intermittent visibility, long-tailed outages, and occasional extended communication gaps, rather than on flight-accurate orbital geometry down to the meter. Future studies can readily replace this approximation with a full CR3BP propagator or with mission-specific ephemerides generated using NASA's General Mission Analysis Tool (GMAT), without requiring any modification to the proposed KCDT framework. Moreover, the Longjiang-2 comparison is particularly valuable because it evaluates the framework using a real spacecraft orbit and observed outage statistics, which cannot be obtained from a purely simulated Gateway model.

A further limitation concerns two distinct dependence mechanisms. First,
multiple estimates of the same link can share observation errors, as in the
direct/moonbounce pair and four-station example of
Section~\ref{sec:longjiang}. Eqs.~(\ref{eq:correlatedfusioncov}) through
(\ref{eq:correlatedfusionmean}) address this case when the joint covariance
is known, while Eqs.~(\ref{eq:cicovariance}) through (\ref{eq:cimean}) provide a
conservative alternative when it is not. Second, the physical-model and
learned-residual errors can be correlated both within and across links, as
formalized in Eqs.~(\ref{eq:hybridmodelerror}) through
(\ref{eq:networkmodelcov}). Shared ephemerides, spacecraft clocks,
transmitter parameters, environmental conditions, model approximations,
neural-network parameters, and observation histories can produce nonzero
cross-link covariance blocks.

The evaluated per-link covariance prediction in
Eq.~(\ref{eq:predictioncov}) does not explicitly propagate these blocks.
Omitted within-link cross-terms can miscalibrate an individual posterior,
whereas omitted cross-link dependence can misstate the joint uncertainty
used by routing, handover, and coupled resource-allocation decisions.
Estimating these covariance blocks online, validating them with synchronized
flight data, and evaluating sparse, low-rank, or grouped correlation-aware
implementations remain open problems. The linear-scaling result in
Fig.~\ref{fig7a} therefore applies only to the evaluated factorized
implementation and should not be extrapolated to a fully joint filter.

Beyond covariance modeling, the practical deployment
of KCDT also requires balancing continuous knowledge maintenance
against onboard energy constraints. Although Fig.~\ref{fig7a}
demonstrates linear computational scaling for the factorized
implementation, energy consumption is not explicitly evaluated
in this work. Future implementations could employ energy-aware,
event-triggered reasoning, using lightweight belief prediction
between observations and activating computationally intensive
inference when uncertainty or evidence age approaches
mission-defined limits. Such adaptive reasoning could reduce
unnecessary computational expenditure while maintaining the
operational awareness required for autonomous cislunar
communication.

Closely related to this energy--computation trade-off
is the amount and granularity of knowledge that the twin should
maintain. Incorporating increasingly detailed physical and network
information may improve modeling fidelity but also increases
computational and storage requirements. Rather than maintaining
an exhaustive representation, the proposed KCDT employs a compact
three-dimensional link state comprising SNR, Doppler shift, and
propagation delay, together with their posterior uncertainty.
Mission-dependent weights and entropy--freshness requirements
further distinguish the knowledge needed for critical links from
that required for less demanding communication tasks. Nevertheless,
the present framework does not explicitly optimize knowledge
selection or representation dimensionality. Future work could
investigate adaptive knowledge granularity and
value-of-information-based selection, incorporating additional
physical or network features only when their expected operational
benefit justifies the associated computational, storage, and
energy costs.

Beyond these modeling and resource considerations,
the validation of KCDT against actual spacecraft observations
introduces additional limitations. The Longjiang-2 comparison
carries its own limitations: the spacecraft's exact hardware
parameters were unavailable, so the link-budget figures are
estimates rather than values obtained from the flight hardware,
and the moonbounce geometry uses a simplified single-point
reflection that captures the two-path structure without
reproducing the reported Doppler near-cancellation. These
checks therefore establish physical plausibility rather than
an exact reproduction of the mission data. In particular,
the common-error fractions in Fig.~\ref{fig10} are controlled
sensitivity parameters, not correlations estimated from
Longjiang-2 measurements. Flight telemetry containing
synchronized station observations would be required to
estimate the joint covariance and validate absolute calibration.

\section{CONCLUSION}

The framework is a starting point rather than a complete solution. Formalizing Knowledge Entropy and Knowledge Freshness through provably convergent update rules, extending the single-node reasoning developed here to distributed fusion across asynchronously operating agents, and establishing the reliability of a foundation-model-based reasoning layer under the same delay and incompleteness are open problems exposed by this framework. More broadly, we expect the shift from state-centric to knowledge-centric reasoning to apply beyond Earth to Moon communications wherever autonomous systems, including deep-sea exploration platforms, disaster-response networks, and distributed robotic fleets, must act on what they currently know rather than on what may eventually become known. We hope this work opens a research direction in which the representation, quantification, and lifecycle of operational knowledge become as central to communication theory as the traditional concept of the bit.

In summary, we introduced a knowledge-centric communication framework for autonomous cislunar networks, in which a digital twin continuously fuses delayed observations, communication physics, learned models, uncertainty quantification, and mission objectives into an evolving Operational Knowledge State, quantified through Knowledge Entropy and Knowledge Freshness. We stated the standard Gaussian-fusion and covariance-aging consistency properties under explicit assumptions, extended the formulation to correlated estimates, and demonstrated their operational implications through Monte Carlo simulations of Earth-orbit and Lunar-Gateway communication regimes, a mission-level estimator comparison against state-centric, physics-only, and data-driven baselines, and a physically grounded comparison based on Longjiang-2 flight data.

\section*{Code Availability}
Code for the proposed KCDT is publicly available on GitHub at \url{https://github.com/afanali85-tech/Knowledge-Centric-Digital-Twin-KCDT-/releases/tag/v1.0.1} and has been archived at \url{https://doi.org/10.5281/zenodo.21699282}.

\section*{Acknowledgment}
Afan Ali and Daniel Benevides da Costa initiated and conceptualized the original draft. Afan Ali designed the methodology, investigation, and validation of the manuscript. Daniel Benevides da Costa and Ali Arshad Nasir provided valuable supervision, suggestions, and proofreading of the paper. The authors declare that they have no competing interests. The work of Daniel Benevides da Costa has been supported by the RDIA Chair Professorship Grant under no. CUP 2601.

\bibliographystyle{IEEEtran}
\bibliography{Ref.bib}

\begin{IEEEbiography}[{\includegraphics[width=1in,height=1.25in,clip,keepaspectratio]{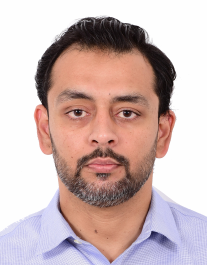}}]{AFAN ALI } (Senior Member, IEEE) Afan Ali received the B.Sc. degree in telecommunications engineering from the National University of Computer and Emerging Sciences (NUCES), Islamabad, Pakistan, in 2008, the M.Sc. degree in engineering from Australian National University (ANU), Canberra, Australia, in 2011, and Ph.D. degree in information and communication engineering from the Northwestern Polytechincal University (NWPU), Xi’an, China, in 2018. From 2009 to 2024, he served at the National Engineering and Scientific Commission (NESCOM), Islamabad, Pakistan, where he progressed through successive research leadership roles. From 2024 to 2025, he was a Postdoctoral Research Fellow in Communications, Signal Processing, and Networking Center (CoSiNC), at the Istanbul Medipol University, Istanbul, Turkey. He received the Best Paper Award 2016 at IEEE Advanced Information Management, Communicates, Electronic and Automation Control Conference (IMCEC), Xi’an, China. He was the recipient of China Outstanding International Student Award 2018, Shaanxi Province, China. He is currently a Postdoctoral Research Fellow at Interdisciplinary Research Center for Communication Systems and Sensing (IRC-CSS), King Fahd University of Petroleum \& Minerals (KFUPM), Saudi Arabia. His research lies at the intersection of wireless communications, active and predictive digital twins, and Generative and Agentic AI for 6G networks. 
\end{IEEEbiography}

\begin{IEEEbiography}[{\includegraphics[width=1in,height=1.25in,clip,keepaspectratio]{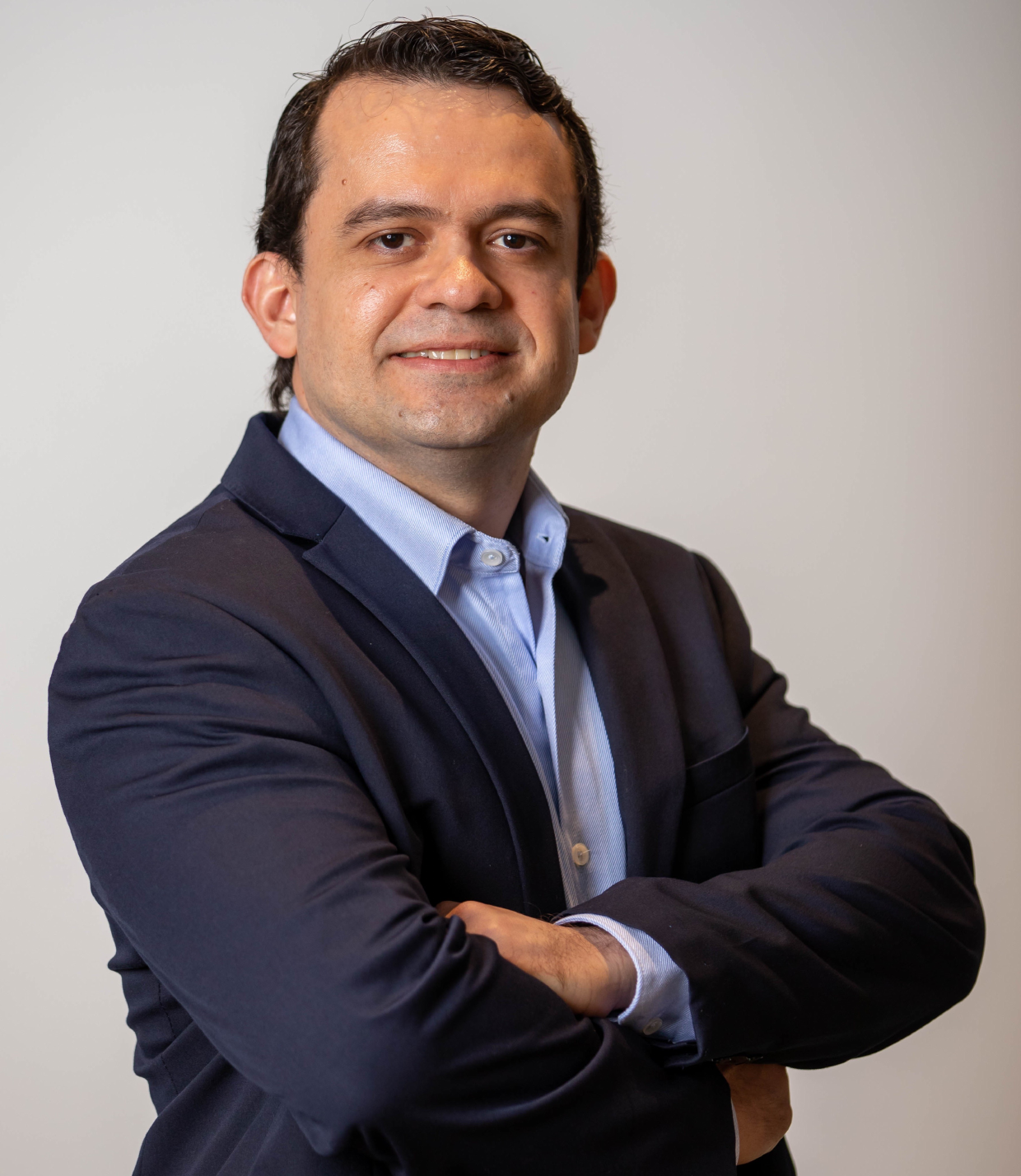}}]{PROF. DANIEL BENEVIDES DA COSTA } (Senior Member, IEEE) Prof. Daniel B. da Costa received the B.Sc. degree in Telecommunications from the Military Institute of Engineering (IME), Rio de Janeiro, Brazil, in 2003, and the M.Sc. and Ph.D. degrees in Electrical Engineering, Area: Telecommunications, from the University of Campinas, SP, Brazil, in 2006 and 2008, respectively. His Ph.D. thesis was awarded the Best Ph.D. Thesis in Electrical Engineering by the Brazilian Ministry of Education (CAPES) at the 2009 CAPES Thesis Contest. From 2010 to 2022, he was Professor at the Federal University of Ceará, Brazil. He also held postdoctoral and visiting positions at INRS-EMT (Canada), Lappeenranta University of Technology (Finland), King Abdullah University of Science and Technology (Saudi Arabia), and Istanbul Medipol University (Turkey), and he was Full Professor at National Yunlin University of Science and Technology (Taiwan). From 2022 to 2024, he was Principal Researcher at the Technology Innovation Institute (TII), Abu Dhabi. He is currently RDIA Chair Professor of Wireless Communications at King Fahd University of Petroleum \& Minerals (KFUPM), Saudi Arabia, and Acting Director of the Interdisciplinary Research Center for Communication Systems and Sensing (IRC-CSS) at KFUPM. Prof. da Costa also serves as Editor-in-Chief of IEEE Communications Letters and Specialty Chief Editor of Frontiers in Communications and Networks, and has held several leadership roles within IEEE societies. 
\end{IEEEbiography}

\begin{IEEEbiography}[{\includegraphics[width=1in,height=1.25in,clip,keepaspectratio]{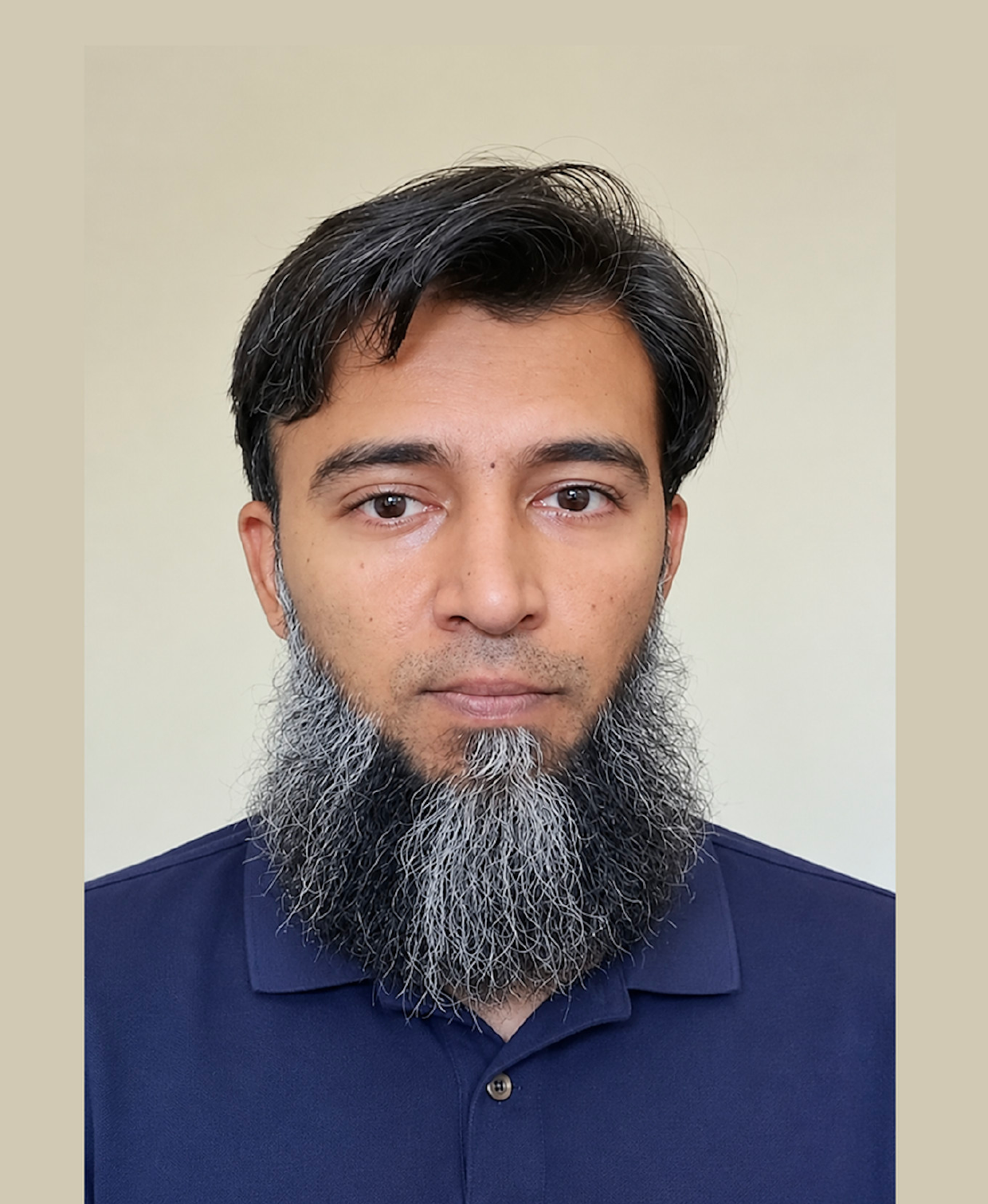}}]{PROF. ALI ARSHAD NASIR } (Senior Member, IEEE) Prof. Ali Arshad Nasir received the Ph.D. degree in telecommunications engineering from the Australian National University, Australia, in 2013, where he worked as a Research Fellow from 2012 to 2015. From 2015 to 2016, he was with the School of Electrical Engineering and Computer Science, National University of Sciences and Technology, Pakistan. He joined the Department of Electrical Engineering, King Fahd University of Petroleum and Minerals, Dhahran, Saudi Arabia, in 2016, where he is currently working as a Professor. His research interests are in the area of signal processing in wireless communication systems. He served as an Editor for IEEE Wireless Communications Letters from 2021 to 2023 and for IEEE Communications Letters from 2024 to 2025. He received the Exemplary Editor Award from IEEE Communications Letters in 2025. He has been serving as an Editor for IEEE Transactions on Communications since 2026.
\end{IEEEbiography}

\end{document}